\documentclass[format=acmsmall, review=false]{acmart}
\usepackage{acm-ec-20}
\usepackage{booktabs} 
\usepackage[ruled]{algorithm2e} 
\usepackage{amsmath,amssymb,amsfonts, amsthm}
\usepackage{theoremref}
\usepackage{times,graphicx}
\usepackage{mathtools}
\usepackage{tikz}
\usepackage{todonotes}
\usepackage{thmtools}
\usepackage{appendix}
\usetikzlibrary{arrows.meta,automata,quotes}
\usepackage{enumerate}
\theoremstyle{acmplain}
\newtheorem{thm}{Theorem}
\newtheorem{cor}[thm]{Corollary}
\newtheorem{lemma}[thm]{Lemma}
\newtheorem{prop}[thm]{Proposition}
\newtheorem{conj}[thm]{Conjecture}
\newtheorem{defi}[thm]{Definition}

\newtheorem{rem}[thm]{Remark}
\newtheorem{obs}[thm]{Observation}
\newtheorem{exa}[thm]{Example}

\newcommand{\N}{\mathbb{N}}
\newcommand{\Z}{\mathbb{Z}}
\newcommand{\B}{\mathcal{B}}
\renewcommand{\L}{\mathcal{L}}
\newcommand{\R}{\mathcal{R}}
\renewcommand{\S}{\mathcal{S}}
\newcount\Comments  
\newcommand{\kibitz}[2]{\ifnum\Comments=1{\color{#1}{#2}}\fi}

\newcommand{\tb}{\ensuremath{\mathrm{TB}}}
\renewcommand{\U}{\ensuremath{\mathcal{U}}}
\newcommand{\M}{\ensuremath{\mathcal{M}}}
\newcommand{\odd}{\ensuremath{\mathrm{odd}}}
\newcommand{\even}{\ensuremath{\mathrm{even}}}
\newcommand{\oddheap}{\ensuremath{\mathrm{odd}}}
\newcommand{\evenheap}{\ensuremath{\mathrm{even}}}
\newcommand{\one}{Left}
\newcommand{\two}{Right}
\newcommand{\out}{o}
\newcommand \pos[2] {\left(#1,#2\right)}
\newcommand \posm[2] {\left(#1,\hat{#2}\,\right)}

\renewcommand \hat \widehat
\renewcommand{\a}{\ensuremath{\mathcal{A}}}

\newcommand{\ny}[2]{\nu_{#2}(#1)}
\newcommand{\nym}[2]{\widehat\nu_{#2}(#1)}
\newcommand{\ovec}[1]{o(#1)}
\renewcommand{\o}[2]{o_{#2}(#1)}

\DeclarePairedDelimiter\ceil{\lceil}{\rceil}
\DeclarePairedDelimiter\floor{\lfloor}{\rfloor}

\SetAlFnt{\small}
\SetAlCapFnt{\small}
\SetAlCapNameFnt{\small}
\SetAlCapHSkip{0pt}
\IncMargin{-\parindent}

\setcitestyle{acmnumeric}

\title[Discrete Richman-bidding Scoring Games]{Discrete Richman-bidding Scoring Games}

\author{Urban Larsson}
\affiliation{%
  \institution{National University of Singapore}
}

\author{Neel Patel}
\affiliation{%
  \institution{National University of Singapore}
}

\author{Ravi Kant Rai}
\affiliation{%
  \institution{Indian Institute of Technology Bombay}
}

\begin{abstract}
We study zero-sum (combinatorial) games, within the framework of so-called Richman auctions (Lazarus et al. 1996) namely, we modify the alternating play scoring ruleset Cumulative Subtraction (CS) (Cohensius et al. 2019), to a discrete bidding scheme (similar to Develin and Payne  2010). Players bid to move and the player with the highest bid wins the move, and hands over the winning bidding amount to the other player. The new game is dubbed Bidding Cumulative Subtraction (BCS). In so-called unitary games, players remove exactly one item out of a single heap of identical items, until the heap is empty, and their actions contribute to a common score, which increases or decreases by one unit depending on whether the maximizing player won the turn or not. We show that there is a unique bidding equilibrium for a much larger class of games, that generalize standard scoring play in the literature.  
We prove that for all sufficiently large heap sizes, the equilibrium outcomes of unitary BCS are eventually periodic, with period 2, and we show that the periodicity appears at the latest for heaps of sizes quadratic in the total budget.

\end{abstract}

\begin{document}

\maketitle


\section{Introduction}\label{sec:intro}
Suppose that you are involved in a 2-player game, where the goal is to capture more objects than your opponent, but at each stage of play, you may only remove one item. You may remove an object if and only if you win a specially designed auction. Say that there are two objects in the heap (and therefore to win the game you need to collect both, and if you collect one of them the game is drawn), and assume that you and your opponent share \$5; you have \$1 and he has \$4. In addition to your single dollar, you have a tie-breaking marker, which works in your favor. If a player wins the bid strictly, by bidding say $d$ dollars, then they get an item but must pay the bidding amount to the opponent. If the bid ends up equal, then the player with the tie breaking marker wins the round, gets an item, but has to hand over \$$d$ and the marker to the other player. Clearly, you cannot win both bids with a single dollar. The question is whether you will be able to tie this game, and win one of the bids, or if the opponent can win both, and the answer is easy. Since the opponent has to bid \$2 to win the bid, then you will win the next round. (If you had only the marker but no dollar you would lose both bids, and if you had only a dollar but no marker you would lose both bids.) From this gentle introduction, we already have an intuition that the marker has a non-trivial impact on such games, and it could be worth sometimes a dollar. The tie-breaking is necessary to play the game, and to understand such games, it seems an unavoidable issue to resolve the worth of this `marker' in {\em every} game situation. We will formalize this `every' setting and present some further major issues to solve. But first, we provide some more background. 


Richman auctions \cite{LLPU1996} are designed for any standard combinatorial 2-player game \cite{BCG2001}, to resolve who is to play next. 
Instead of alternating play, for each stage of game, the 2 players, called  Left and Right,\footnote{These names are adapted from standard literature on combinatorial games, Left is `she' and Right is `he'.} resolve this crucial moment by a type of auction where the winning player must pay the losing player their bid amount. 

 Richman bidding can be adapted to any standard combinatorial game, and thus offers a way to extend classical Combinatorial Game Theory (CGT) \cite{S2013,BCG2001,C1976} to a more economic style of gameplay. 
 Moreover, Richman auctions are a perfect fit for combinatorial games, as they offer an unlimited number of bidding rounds. This is useful, since many popular board games can be played on arbitrarily large game boards, and various kinds of  Nim-type removal games (a.k.a. heap games, take-away games, etc)  \cite{BCG2001} can contain arbitrarily many objects in a starting position. 
 
 
    
 In this paper, we study discrete bidding in a manner similar to \cite{DP2010}. We adapt the setting to a class of games known as {\em Cumulative Subtraction} (CS) \cite{CLMW2020}. 
 CS is played on a finite heap of pebbles, and players take turns removing pebbles from the heap, but are only allowed to remove a number of pebbles from a fixed set of values, called the {\em subtraction set} $S\subset \N$. For example, if $S = \{3,5,9\}$, then players are allowed to remove either $3$, $5$ or $9$ pebbles at each turn, as long as the heap size stays non-negative. 
 In the standard zero-sum variation, the final score is the difference between the number of pebbles accumulated by each player by the end of the game. 
 Player Left is the maximizer, whereas player Right is the minimizer; the final score is the total number of pebbles collected by Left during play minus Right's total number. In order not to obscure the main ideas of our bidding setting, we will simplify the `subtraction-part' of the games, and focus on the simplest possible subtraction set, namely when $S=\{1\}$ (although several results hold also for more general subtraction sets). 
 
 Let $\N=\{1,2,\ldots\}$, and let $\N_0=\N\cup\{0\}$. Discrete bidding here means that there is a given total budget of $\tb\in \N_0$. One of the players has a budget $p\in \{0,\ldots ,\tb\}$, and the other player has the budget $q = \tb - p$. 
 At each turn, players submit a (closed) bid. The player with the higher bid gets to play, and pays their bid to the player with the lower bid. 
 Thus, if player Left wins by bidding $\ell$, she pays this bid to Right, and makes her desired move. The new budget partition becomes $(p-\ell, q+\ell)$. Ties are resolved using a {\em tie-breaking marker}: one of the players has the marker, and this player wins the turn in case of equal bids. They makes their desired move, pays the other player the bidding amount, and passes them the marker.
 
By adapting the bidding mechanism to CS, we get a new game, called {\em Bidding Cumulative Subtraction} (BCS). 


The discrete setting gives new challenges and problems to Richman games. 
For example, the marker, which is a necessary and convenient tool to resolve a tie, can sometimes be worth a dollar (but not more).\footnote{In the continuous setting, any positive bidding amount is worth more than the marker.} Perhaps this seems like an unfortunate side effect of the discrete setting, but it turns out that it is quite interesting to resolve the accompanying asymmetry of the game. In many situations, the player with the marker has a slight advantage, and it appears non-trivial to find out exactly what that means, even in the case of the simplest subtraction set, $S = \{1\}$, which will also be called the \emph{unitary} setting.  


 
\subsection{Related Work}\label{sec:related}

In the classical Richman setting \cite{LLPU1996} \cite{LLPSU1999}, the auction is continuous, say with total budget \$1, and the players split that budget to, say $p$ and $q$, with $p+q=1$. Optimal bids have been resolved for the game {\sc tic-tac-toe} \cite{DP2010}, but very little is known for the game of {\sc chess} \cite{LW2019}. 

The main theorem of Richman games considers one additional setting, where the turn function is a Bernoulli trial: for each stage of play, an $r$-biased coin is tossed to decide whose turn it is. A player wins a given coin tossing game if and only if the player with budget $\$(1-r)$ wins the corresponding Richman game. A rule to resolve ties is required, but since ties occur with probability 0, in the continuous setting, the main result holds for any standard tie breaking rule, with open or sealed (simultaneous) bids. In games with zugzwangs (games where no player wants to move first), it is customary to let the winner of a bid decide who is to play next. In games with a nonnegative incentive (such as the games we will study), it suffices to have the players  bidding for the opportunity to take the next turn.

\citeauthor{DP2010} \cite{DP2010} study win-loss games, with discrete Richman bidding, and they use a modified tie resolution method: the player with the marker may either keep the marker and let the other player pay the bid and move in the game, or give away the marker together with the bid and make their desired move in the game. 
This is not just a subtle difference in approach; namely, we emphasize that our simpler approach has the benefit of generalizing alternating play. Namely, in our setting, alternating play corresponds to $\tb=0$ (with both players having a respective individual budget of $0$ as a result). Another obvious difference is that we study games with numeric results, in contrast to the traditional CGT win-loss situation. This leads to an interesting question about uniqueness of bidding equilibrium. 

Combinatorial games are often viewed in the light of so-called normal play, where it is always good to be able to move (an idea, which recently leads to an `absolute combinatorial game theory' \cite{LNS,LNS2017}). Here, we may have a situation where the player who wins the bid is the normal play loser. Thus, for our most general setting, we are motivated to include a certain penalty function, $\tau$, with input Left and Right terminal positions. If a player wins a bid that has no move, there is some consequence of this final auction; indeed in Guaranteed Scoring games \cite{LNS2018,LNNS2016} a `terminal penalty' is invoked for a player who cannot move (and this penalty corresponds to the maximum score that the other player can still achieve by playing out all remaining move options). We include this short discussion here, since we want to emphasize that the proposed move convention in this work generalizes game settings in the literature. 

All literature on combinatorial games assumes a unique equilibrium, which defines a game value under \emph{optimal play}; see for example \cite{M1953, H1959, E1996, LNS2018, LNNS2016, J2014, S2013}.

\subsection{Our Contribution}\label{sec:contrib}


In this paper, we first analyze a general setting of discrete Richman bidding adapted for any standard combinatorial scoring game. We show that if the ruleset satisfies some basic and intuitive monotonicity properties, then each such game exhibits a unique equilibrium. Then in the rest of the paper, we focus on a simple class of Bidding Cumulative Games (defined in Section~\ref{sec:results}) named {\em unitary games}. Our theoretical contribution can be summarized in three main results.

In Section~\ref{sec:unigames}, we show that every unitary game has a unique equilibrium described recursively by maximin functions of the game positions  (Theorem~\ref{thm:unitary}). We prove this result by showing that the ruleset of unitary games satisfies sufficient properties for uniqueness. Then we prove further monotonicity properties of unitary games, which helps us to prove a convergence result. In Section~\ref{sec:bounded}, we prove that the equilibrium outcomes of unitary games converge for heap sizes of the same parity, i.e. the sequence is eventually periodic with period 2  (Theorem~\ref{thm:const}). Moreover, aided by a certain `bidding automaton', in Section~\ref{sec:automaton}, we show that the convergence to this periodicity is quadratic in the size of the total budget (Theorem~\ref{thm:convbou}). In addition, we provide a conjecture on the corresponding closed formula expression for the equilibrium outcomes in the limit (Conjecture~\ref{conj:conv}).        

\section{Equilibrium properties}\label{sec:propU}

Let us first elaborate on the fundamental properties of discrete Richman bidding, for any combinatorial scoring  (zero-sum) game. 
For a given total budget $\tb\in\N_0$, let $\tb=p+q$, where $p$ is Left's part of the budget, and unless otherwise stated, Left has the tie-breaking marker. 
Moreover, let $\varnothing\ne\B\subseteq \{0,\ldots \tb\}$, be the set of possible Richman bids.\footnote{We let the bidding set be symmetric to simplify notation.} 
A player may not be able to bid on their turn; for example, if $0\not\in \B$ and the player's budget has been exhausted. In this case, the player who is able to place a bid moves and transfers a valid bid amount $b\in\B$ to the other player. In case there is a position such that no player can bid, the ruleset is invalid. We will assume valid rulesets, i.e. {\em some} player is able to place a bid at every turn.

In this section, we consider acyclic combinatorial rulesets of the form $\Gamma=(X, \L, \R, \B)$, where $X$ is a set of positions (nodes), and $\L,\R: X\rightarrow 2^X$ are the move functions for players Left and Right respectively, with discrete Richman $\B$-bidding, and with given weight functions $w_L, w_R:X\times X\rightarrow \mathbb R$ on the set of move edges, for each player. 
Here $X$ denotes a possibly infinite set of `starting' positions, in a usual sense of CGT: any game state may be considered as a starting position, but each play sequence is finite. A Left move is of the form $x\in X\rightarrow y\in \L(x)$, and similarly for Right. The final score (or utility) of a terminating play sequence $\sigma$ is 

$$u(\sigma) = \tau(t)+\sum w_L(e_L)-w_R(e_R),$$ where \one\ played the moves $e_L$ and \two\ played the moves $e_R$. 
The weighted move edges are chosen at each stage of play by the winner of the Richman bidding. The terminal \emph{penalty scores} are given by a function $\tau: T_L\cup T_R\rightarrow \R$, where $T_L,T_R\subset X$ is the set of terminal positions for Left and Right respectively.  
The game ends when the player who wins the Richman bidding cannot move.

Throughout the paper, given a total budget $\tb$, for a given position $x\in X$, a \emph{Richman-position} is denoted by $(x, \hat p)$, where Left has $p$ dollars and the tie-break marker, or by $(x, p)$, where Left has $p$ dollars but Right has the marker. A note on terminology: we will abbreviate  `Richman-position' to \emph{position}, since the generalized CGT `move-flag', the current budget partition (sometimes denoted simply by $\hat p$ or $p$), is mandatory information to play.\footnote{This is consistent with standard CGT terminology, which often omits the notion of who is to move, since the theory requires analysis of all players as starting players.} 

A player may have move options, but not being able to access them even if the other player does not have any move options. For example, if Left has no remaining budget, and Right has the marker, then if Left is the only player with move options, the game will end, because Right will win that bid.  
Combinatorial games are often viewed in the light of so-called normal play, where it is always good to be able to move. 
Here, we may have a situation where the player who wins the bid is the normal play loser. 

 This is the motivation for the penalty function, $\tau$; if a player wins a bid that has no move, there is some consequence of this final auction.

For symmetric games, where $\L=\R$ and $w_L=w_R$, this discussion is obsolete, and we may set $\tau(t)=0$, for all terminal positions $t$. 
In the symmetric case, there is no bidding when one of the players runs out of options, because then both players run out of options, and there is no possibility that the game may continue. 

Our games generalize standard combinatorial scoring games; namely, by taking $\tb=0$, we obtain the standard alternating play mechanics. 
 As mentioned, combinatorial games have unique equilibria,  which defines game values under \emph{optimal play}, and depending on who starts the game. 
We will demonstrate that we do have uniqueness of equilibrium in our generalization, under a natural set of axioms. In those cases, we may refer to unique \emph{game values}. 

Let us first define the relevant maximin functions. 
The maximin function induces a bid-action pair that offers the greatest value of a set of minima: Left declares bids and notes Right's bid-action responses in each case. Then Left chooses the bid that maximizes the value, and given that Left can choose action if she wins the bid. Minimax is the reverse situation.

\begin{defi}[Maximin  Functions]\thlabel{def:maximin}
Consider a  a total budget $\tb\in \N_0$ and a ruleset $(X, \L, \R,\B)$. The maximin functions $\widehat\nu,\nu: X\rightarrow \mathbb R^{\tb+1}$ are defined on positions $(x, \hat p)$ and $(x, p)$, respectively.

For all $p\in\{0,\ldots , \tb\}$, for terminal positions $t_L, t_R\in X$,  $\nym{t_L}{p}=\tau(t_L)$, $\ny{t_R}{p}=  \tau(t_R)$, and recursively, for non-terminal positions $x\in X$, 
$$\nym{x}{p} = \max_{\ell, y\in \mathcal L(x)}\; \min_{ r,  z\in \mathcal R(x)} \{\nym{y}{p-\ell}|_{\ell>r}+w_L(x,y), \nu_{p-\ell}(y)|_{\ell = r}+w_L(x,y), \nym{z}{p+r}|_{\ell<r}+w_R(x,z)\}$$
and
$$\ny{x}{p} = \max_{\ell, y\in \L(x)}\;\min_{ r, z\in \mathcal \R(x)}  \{\ny{y}{p-\ell}|_{\ell>r}+w_L(x,y), \nym{z}{p+r}|_{\ell = r}+w_R(x,z), \ny{z}{p+r}|_{\ell<r}+w_R(x,z)\},$$
where $\ell\in \B\cap \{0,\ldots, p\}$ and $r\in \B\cap \{0,\ldots, q\}$.
\end{defi}
The minimax functions are  defined analogously. 

The classical `CGT outcome' in alternating play from position $x$ corresponds to $\tb=0$, and it is the ordered pair $o(x)=(\nym{x}{0},\ny{x}{0})$, the optimal play score when Left and Right starts, respectively. 
Note that, since the bidding is trivial, for this case, there is a unique equilibrium. 
In general, if, for all positions $x$, each entry has a unique equilibrium, the generalized CGT-outcome, for $\tb\in\N$, is the vector $\out (x)= (\nym{x}{\tb},\ldots, \nym{x}{0},\ny{x}{0},\ldots ,\ny{x}{\tb})$. 
Whenever applicable, we call this vector $ \out (x)$ the (optimal play) {\em outcome} of $x$. For symmetric games, i.e. when the move sets $\L = \R$ are the same, it suffices to store half this vector, and so we drop the marker notation, and write $\ovec{x}=(\nym{x}{\tb},\ldots, \nym{x}{0})$ (i.e. symmetric outcomes assume that Left has the marker).

In Definition~\ref{def:U}, we provide intuitive properties of the ruleset $\Gamma$ and  
 we prove that the following properties suffice to establish uniqueness of equilibria in Theorem~\ref{thm:unique}. 
\begin{defi}[Uniqueness Properties]\thlabel{def:U}
 A ruleset $\Gamma\in \mathcal U$ if the pair of maximin functions $(\hat\nu$, $\nu)$ satisfies, for each position $x$, for all $p\in\{0,\ldots , \tb\}$:
\begin{itemize}
\item[(A)] `Budget Monotonicity.' $\nym{x}{p}\ge \nym{x}{\pi}$ and $\ny{x}{p}\ge \ny{x}{\pi}$, if $p>\pi$. 
\item[(B)] `Marker Monotonicity.' $\nym{x}{p}\ge \ny{x}{p}$.
\item[(C)] `Marker Worth.' $\nym{x}{p}\le \ny{x}{p+1}$.
\end{itemize}
\end{defi}
If a ruleset $\Gamma\in \U$, then, by (A) the players weakly prefer to win by smaller bids, by (B) Right weakly prefers to tie a Left winning bid, and by (C), given a Right bid, Left weakly prefers to tie before winning strictly (the marker is worth at most a dollar). These properties are fairly straight forward and intuitive in our setting. We now show in Theorem~\ref{thm:unique} that if ruleset $\Gamma$ satisfies such properties, then it exhibits unique equilibria.   

%
%
%
%
%
%
\begin{thm}[Uniqueness]\thlabel{thm:unique}
If a ruleset $\Gamma \in\mathcal U$, then for all positions and for all budget partitions, the game has a unique equilibrium.
\end{thm}

\begin{proof} 
We assume with no loss of generality that Left has the marker: a symmetric argument holds when Right has the marker. 

Suppose that Left bids `0'. Then, if Right benefits by deviating, by (A) and the zero-sum property, we may assume that he bids `1' (perhaps he bids more if the same value). Suppose next that Left bids 1. Then similarly Right might benefit by bidding 2, and so on.

Item (C) implies that Left does not benefit by winning strictly at bid $\ell$, if she can win by a tie at bid $\ell-1$. Item (B) implies that Left does not prefer a Right tie at a given Left bid. By combining these two statements, Left weakly prefers a tie bid at $\ell\ge 0$, before any larger tie. 

Altogether, this deviation scheme must terminate at some tie bid $\ell\le p$ or possibly Right wins at $\ell+1\le q$, and in either case the unique equilibrium is given by Definition~\ref{def:maximin}.
\end{proof}


For the special case of symmetric games, we simplify the notations, and  denote $\L=\R=\mathcal M$. In Observation~\ref{obs:maximinsym}, we define maximin function for symmetric games which is a special case of the maximin function defined in the Definition~\ref{def:maximin}. 
\begin{prop}[Symmetric Maximin]\thlabel{obs:maximinsym}
Consider a symetric ruleset. For all terminal positions $t\in X$, for all $p$, $\nym{t}{p}=\ny{t}{p}=0$. 
For non-terminal $x$,
$$\nym{x}{p} = \max_{0\le \ell\le p, y\in \mathcal M(x)}\; \min_{ 0\le r\le q,  y\in \mathcal M(x)} \{\hat\nu_{p-\ell}(y)|_{\ell>r}+w(x,y), \nu_{p-\ell}(y)|_{\ell = r}+w(x,y), \hat\nu_{p+r}(y)|_{\ell<r}-w(x,y)\},$$
where, for all $x$, $p$,  
\begin{align}\label{eq:0-sum}
\ny{x}{p}= -\nym{x}{q}.  
\end{align}
\end{prop}
\begin{proof}
We express the equality \eqref{eq:0-sum} as
\begin{align*}
    \ny{x}{p}&=-\nym{x}{q} \\
    &= -\max_{0\le \ell\le q}\; \min_{ 0\le r\le p} \{\nym{y}{q-\ell}|_{\ell>r}+w, \ny{y}{q-\ell}|_{\ell = r}+w, \nym{y}{q+r}|_{\ell<r}-w\}\\
    &=-\max_{0\le \ell\le q}\; \min_{ 0\le r\le p} \{-\ny{y}{p+\ell}|_{\ell>r}+w, -\nym{y}{p+\ell}|_{\ell = r}+w, -\ny{y}{p-r}|_{\ell<r}-w\}\\
    &=\max_{0\le \ell\le q}\; \min_{ 0\le r\le p} \{\ny{y}{p+r}|_{r>\ell}+w, \nym{y}{p+r}|_{\ell = r}+w, \ny{y}{p-\ell}|_{r<\ell}-w\}
\end{align*}
where, at each instance, $y\in\M(x)$, and where the middle equality is by \eqref{eq:0-sum}. Since the maximin is now applied on the negative maximin values, we must swap the players inside the brackets, and hence the last equality follows when we remove the signs. This coincides with \thref{def:maximin}, whenever $\L=\R=\M$.
\end{proof}

And even simpler, by the restriction to symmetric games, and in view of \eqref{eq:0-sum}, we may leave out the function $\nu$. 
\begin{obs}[Simplified Maximin]\thlabel{obs:maximinsim}
The tuple of maximin functions $\hat\nu: X\times \{0,1\}\rightarrow \mathbb R^{\tb+1}$ is, for all $p\in\{0,\ldots , \tb\}$, if $t\in X$ is terminal, $\nym{t}{p}=0$, and, for non-terminal $x$, given $\ell, r\in \B$,
$$\nym{x}{p} = \max_{0\le \ell\le p, y\in \mathcal M(x)}\; \min_{ 0\le r\le q,  y\in \mathcal M(x)} \{\hat\nu_{p-\ell}(y)|_{\ell>r}+w(x,y), -\nym{x}{q+\ell}|_{\ell = r}+w(x,y), \hat\nu_{p+r}(y)|_{\ell<r}-w(x,y)\}.$$

\end{obs}

If a ruleset satisfies the uniqueness properties, by the proof of \thref{thm:unique}, we may simplify the maximin function a bit further. For example, if the game is symmetric, we get the following convenient simplification of Observation~\ref{obs:maximinsim}. In Corollary~\ref{cor:U}, we show a simpler form of maximin function for symmetric games whenever ruleset $\Gamma \in \U$.

\begin{cor}\thlabel{cor:U} 
If the ruleset $\Gamma \in\U$ is symmetric, then for all non-terminal game positions and for all budget partitions, the unique equilibrium value is given by: for all $p\in\{0,\ldots , \tb\}$, if $t\in X$ is terminal, $\nym{t}{p}=0$, and, for non-terminal $x$, given $\ell, r\in \B$,
\begin{align}\label{eq:U}
 \nym{x}{p} = \max_{0\le \ell\le p, y\in\M(x)}\; \min_{ 0\le r\le q, z\in \M(x)} \{-\nym{y}{q+\ell}|_{\ell = r}+w(x,y),\nym{z}{p+r}|_{\ell<r\le q}-w(x,z)\}.
\end{align}
\end{cor}
\begin{proof}
This follows by combining Observation~\ref{obs:maximinsim} with \thref{thm:unique}, since given any Right bid, by (C), Left prefers a tie, before winning strictly.
\end{proof}
This simplifies the analysis of equilibrium, because if a ruleset satisfies $\U$ then the cases where Left wins strictly need not be considered because Right always prefers to tie a Left winning bid. In general, we have the following conjecture.

\begin{conj}
Consider a ruleset $\Gamma$. If all weights are nonnegative for Left and all weights nonpositive for Right, and $0\in\B$, then $\Gamma\in U$.
\end{conj}
Moreover, in the following examples, we show that if a ruleset has negative Left-weights and/or positive Right-weights, then there exist zugzwang games for which property $\U$ does not hold. However, when $0\in \B$, we have not found any game that violates property $\U(A)$ or $\U(C)$.  



\begin{example}\thlabel{ex:B}
Consider a ruleset $\Gamma$, where $X=\{x_1,x_2\}$, $\R(x_1)=\{x_2\}$, $\R(x_1) = \L(x_1) = \L(x_2)=\varnothing$ and $0\in\B$. Suppose that $w_R(x_1,x_2)=1$, and $\tau(x_2)=0$, i.e. there is no penalty for not being able to move from the terminal position $x_2$. If the current score $c(x_1)=0$ and Right has the marker but no budget, then if Left bids 0 at position $x_1$, Right must move to the terminal position $x_2$, and the game ends at the final score $c(x_2)=1$, which is good for Left. If Right does not have the marker (independently of the budget partition), he is better off, because he will bid 0, and the game will end at $c(x_1)=0$ (if there were a penalty at this terminal that would have been a Left penalty, but Right would not have been able to take benefit). Thus, $0=\nym{x_1}{1}<\ny{x_1}{1}=1$, and property \U(B) is not satisfied.
\end{example}

\begin{example}
Suppose that $0\not\in\B$, and otherwise with $\Gamma$ as in \thref{ex:B}. Analogously, a player does not want to have any budget, and thus \U(A) may not hold. Similarly, one can see that the marker can be worth more than a dollar; if $\tb=1$ and you have the marker but not the dollar, then if you exchange the marker for the dollar, you are worse off. Hence \U(C) may not be satisfied if $0\not\in\B$.
\end{example}


\section{Bidding Cumulative Subtraction}\label{sec:results}

 In this section, games are symmetric, i.e. the move options are the same for both players; see \cite{CLMW2020} for the motivation on similar alternating play games. We now define the ruleset Bidding Cumulative Subtraction.  

\begin{defi}[Bidding Cumulative Subtraction, BCS]\label{def:game}
There is a subtraction set $\S\subset \N$, a total budget $\tb\in \N_0$, a bidding set $\B\subset \{0,\ldots \tb\}$, and a heap of finitely many, $x\in \N_0$, objects (pebbles). There are two players, Left and Right, who take turns removing objects from the heap. The total budget $\tb\in\N_0$ (a game constant) is partitioned between the players, as $(p,q)$, with $p + q = \tb$. Exactly one of the players has a tie-break marker $m\in \{0,1\}$, where $m=1$ if Left has the marker.  A complete game configuration is of the form $(\S, \B; x,p,m,c)$, where $c\in \Z$ is the current score. If Left has the marker and $c=0$, we abbreviate a position by $(x,\hat p)$, and otherwise, when Right has the marker, we write $(x,p)$. At each stage of play, the players (make closed) bid of who is to take an action, Left bids $\ell$ and Right bids $r$. The player with the highest bid wins the move. If the bids are equal, the player with the marker wins the move. The winning bidder transfers the bidding amount (together with the marker in case of a tie) to the other player. If Left has the marker, a typical bid is $(\hat \ell, r)$. A player who wins the bid acts by collecting  $s\in \S$ objects, which adds $s$ to a current score $c$, if Left wins the bid, and otherwise, it subtracts $s$. The game ends when the number of objects in the heap is smaller than $\min S$. Left seeks to maximize the final score (utility) whereas Right seeks to minimize it. A removal of $s\in S$, in case Left has the marker, is of one of the forms:        
\begin{itemize}
    \item $(\tb;x,p,1;c)\rightarrow (\tb;x-s,p-\ell,1;c+s)$, $s\in S$, if Left bids $\ell\le p$ and wins a non-tie. 
    \item $(\tb;x,p,1;c)\rightarrow (\tb;x-s,p-\ell,0;c+s)$, $s\in S$, if Left bids $\ell\le p$ and wins a tie.
    \item $(\tb;x,p,1;c)\rightarrow (\tb;x-s,p+r,1;c-s)$, $s\in S$, if Right wins by bidding $r\le q$. 
\end{itemize}
\end{defi}


In view of Section~\ref{sec:propU}, we have the ruleset $\Gamma= (\N_0, \M )$, where , for all $x\ge \min S$, $\M(x)= \{x-s\ge 0\mid s\in S\}$. In rest of this section, we focus on {\em unitary games}, which is a BCS with subtraction set $S = \{1\}$. We illustrate bidding in such games in Section~\ref{sec:illustation}. 

\subsection{Unitary Games} \label{sec:unigames}

If $\S=\{1\}$ and $\B=\{0,\ldots , \tb\}$, we call BCS \emph{unitary}. In this section, we prove that for each heap size and each budget partition, there is a unique equilibrium outcome, given by the maximin function (\thref{thm:unique}). That is, by using the notation in Section~\ref{sec:propU}, we prove that the symmetric ruleset $\Gamma = (\N_0,\M)\in \U$, if for all $x\in \N, \M(x)=\{x-1\}$. In the spirit of Definition~\ref{def:game}, for unitary games, we write $(\tb;x,p,m;c):=(\{1\},\{0,\ldots ,\tb\};x,p,m;c)$. Unitary games simplify quite a lot and allows us to prove simple properties on the bound game values, uniqueness of equilibrium and convergence bounds. We start with a simple result in Lemma~\ref{lem:pari}.



\begin{lemma}[Parity]\thlabel{lem:pari}
For unitary games, the utility of any sequence of play from a heap is odd if and only if the size of the heap is odd. 
\end{lemma}

\begin{proof}
For an odd heap size $x$, if we divide the total wins between Left and Right, it has to be even for one player and odd for the other.  If $x$ is even, the number of wins for the players will either both be even or both will be odd. Thus, the difference is even. 
\end{proof}

Since unitary games are symmetric, we may assume that Left has the marker, unless otherwise stated, so $m=1$ will be the default. We define maximin function for unitary games in Defination~\ref{def:maximinunitary} which is a special case of Observation~\ref{obs:maximinsim}.

\begin{defi}[Unitary Maximin]\label{def:maximinunitary}
The maximin function $\hat\nu: X\times \{0,1\}\rightarrow \mathbb R^{\tb+1}$ is, for all $p\in\{0,\ldots , \tb\}$, $\nym{0}{p}=0$, and for $x>0$, 
$$\nym{x}{p} = \max_{0\le \ell\le p}\; \min_{ 0\le r\le q} \{\hat\nu_{p-\ell}(x-1)|_{\ell>r}+1, -\nym{x-1}{q+\ell}|_{\ell = r}+1, \hat\nu_{p+r}(x-1)|_{\ell<r}-1\}.$$
\end{defi}

We are now ready to prove the first main result of the paper in Theorem~\ref{thm:unitary}, which establishes the unique equilibrium in unitary games. However, before proving Theorem~\ref{thm:unitary}, we show the following examples to motivate the result, \thref{thm:unitary}, by maximin play in simple unitary games. Let us first revisit the example in the first paragraph, in a more formal manner.  

\begin{exa}[Maximin Play]\thlabel{ex:?}
Let us illustrate Definition~\ref{def:maximinunitary} with an example where $\tb=5$, and $x=2$, displaying the values $\o{x}{p}$ for $\tb =5$ and the heap sizes $x = 0,1,2$. We explain the entry with the question mark in the table.
{\em 
\begin{table}[ht] 
\centering{
\begin{tabular}{|c||c|c|c|c|c|c|}
\hline 
 $x\setminus \hat p$ & 5 & 4 & 3 & 2 & 1 & 0 \\ \hline\hline
0 & 0 & 0 & 0 & 0 & 0 & 0\\ \hline
1 & 1 & 1  & 1 & -1 & -1 & -1 \\ \hline
2 & 2 & 2  & 0 & 0 & `?' &  -2\\ \hline
\end{tabular}}
\end{table}
}
If the position is $(2, \hat 1)$, then $\ell\in\{0,1\}$ and $r\in \{0,1,2\}$, by the reduced form equivalent. If $\ell = 0$, then Rights minimum is $\min\{-\nym{1}{4+0}+1,\nym{1}{1+1}-1\} = -2$. If $\ell = 1$, then Right's minimum is $\min\{\nym{1}{1-1}+1, \ny{1}{1-1}+1, \nym{1}{1+2}-1 \} = \min\{\nym{1}{0}+1,-\nym{1}{5}+1,\nym{1}{3}-1\}=0$. Hence, Left will bid $\ell=1$ and $\nym{2}{1}=\max\{-2,0\}=0$. 
\end{exa}

\begin{exa}\label{ex:tb9}
Consider $\tb=9$. We computed maximin values for heap sizes $x\le 8$, and gain the following bidding tables for heap size $x=9$. Left bids $\ell$ and Right bids $r$. The lower right corners show the equilibrium values, i.e. the maximin $=$ minimax value; see \thref{thm:unique}. Note that everything to the right of the bold is smaller than or equal the diagonal in bold, but by the second table, the bold does not bound the below area (as in the proof of \thref{thm:unique}).\\

\centering{
{\em 
\begin{tabular}{|c|ccccccc|c|}
\hline 
 $r\setminus \ell$ & 0 & 1 & 2 & 3 & 4 & 5& 6 & $\max$\\ \hline
0 & {\bf 1} & 1 & 1 & 1 & -1 & -1&-3& 1\\ 
1 & 1 & {\bf 1}  & 1 & 1 & -1 & -1&-3& 1 \\ 
2 & 3 & 3  & {\bf 1} & 1 & -1 &  -1&-3& 3\\ 
3 & 3 & 3  & 3 & {\bf -1} & -1 & -1&-3& 3 \\ \hline
$\min$&1&1&1&-1&-1&-1&-3& {1}\\ \hline
\end{tabular}\hspace{5mm}
\begin{tabular}{|c|ccccc|c|}
\hline 
 $r\setminus \ell$ & 0 & 1 & 2 & 3 & 4 & $\max$\\ \hline
0 & {\bf 1} & 1 & -1 & -1 &-1 &1\\ 
1 & -1 & {\bf 1} & -1 & -1 & -1&1 \\ 
2 & -1 & -1  & {\bf -1} & -1 & -1 &-1\\ 
3 & 1 & 1  & 1 & {\bf -1} & -1 & 1\\ 
4 & 1 & 1  & 1 & 1& {\bf -3} & 1\\ 
5 & 3 & 3  & 3 & 3 & 3  &3\\ \hline
$\min$&-1&-1&-1&-1&-3& {-1}\\ \hline
\end{tabular}
}
}

\end{exa}

For unitary games, we aim to prove 

\begin{thm}[Equilibrium for Unitary Games]\thlabel{thm:unitary} 
If BCS is unitary, then for all positions, it has a unique equilibrium. It is given by, for all $p$, $\o{0}{p}=0$, and if $x>0$, then
\begin{align}\label{eq:unitaryequil}
\o{x}{p} = \max_{0\le \ell\le p}\; \min_{ 0\le r\le q} \{-\o{x-1}{q+\ell}|_{\ell = r}+1,\o{x-1}{p+r}|_{\ell<r\le q}-1\}.
\end{align}
\end{thm}
As demonstrated by Corollary~\ref{cor:U}, the proof of the theorem will follow, by establishing the properties in Definition~\ref{def:U} for unitary games. In Section~\ref{sec:propunitary}, we prove the properties in Definition~\ref{def:U} for unitary games.

After this, in Section~\ref{sec:bounded}, we prove a certain monotonicity result on increasing heap sizes (\thref{lem:rowsplitmon}). This will lead to the second main result of this paper, an eventual period 2 of equilibrium outcomes in case of unitary games.

In Figures~\ref{fig:even} and \ref{fig:odd}, we show the equilibrium outcomes for all sufficiently large heap sizes for the special cases of $\tb = 8$ and $9$ respectively. This eventual stabilization of the outcomes and bids is later formalized via an independent \emph{bidding automaton}, and we conjecture that these are generic examples.

\begin{figure} [ht!]
\begin{center}

\begin{tikzpicture} [scale = .9]

\foreach \x/\y in {1/0,2/0,3/0, 4/0, 5/0,1/1,2/1,3/1, 4/1, 5/1,6/0,7/0,8/0, 9/0, 6/1,7/1,8/1, 9/1}
 \draw[thick] (\x, \y) rectangle (\x+1, \y+1);

\draw[blue,thick, ->] (1.5, 0.3) ..controls (3.3, -1) and (8.5, -0.3) .. (9.5, 1.3);
\draw (5, -.7) node {``$\boldsymbol{-}$''};

\draw (1.5, 2.7) node {$(x,\hat 8)$};
\draw (2.5, 2.7) node {$(x,\hat 7)$};
\draw (3.5 , 2.7) node {$(x,\hat 6)$};
\draw (4.5, 2.7) node {$(x,\hat 5)$};
\draw (5.5, 2.7) node {$(x,\hat 4)$};
\draw (6.5, 2.7) node {$(x,\hat 3)$};
\draw (7.5, 2.7) node {$(x,\hat 2)$};
\draw (8.5 , 2.7) node {$(x,\hat 1)$};
\draw (9.5, 2.7) node {$(x,\hat 0)$};

\draw (0 , 1.5) node {$x$ even};

\draw (1.5, 0.5) node {$+5$};
\draw (2.5, 0.5) node {$+3$};
\draw (3.5 , 0.5) node {$+3$};
\draw (4.5, 0.5) node {$+1$};
\draw (5.5, 0.5) node {$+1$};
\draw (6.5, 0.5) node {$-1$};
\draw (7.5, 0.5) node {$-1$};
\draw (8.5 , 0.5) node {$-3$};
\draw (9.5, 0.5) node {$-3$};

\draw (0 , 0.5) node {$x$ odd};

\draw (1.5, 1.5) node {$+4$};
\draw (2.5, 1.5) node {$+4$};
\draw (3.5 , 1.5) node {$+2$};
\draw (4.5, 1.5) node {$+2$};
\draw (5.5, 1.5) node {$0$};
\draw (6.5, 1.5) node {$0$};
\draw (7.5, 1.5) node {$-2$};
\draw (8.5 , 1.5) node {$-2$};
\draw (9.5, 1.5) node {$-4$};

\end{tikzpicture}\caption{The value table for  $\tb = 8$ and sufficiently large heap sizes $x$, even and odd respectively. The blue arrow indicates a shift of sign in the representation of $(x,\hat 8)\rightarrow (x-1,8)$, $x$ odd; that is, both players bid 0, and Left wins the bid by tie-breaking. This bidding is equilibrium play, since $5=1-(-4)$.}\label{fig:even}

\vspace{.5 cm}

\begin{tikzpicture} [scale = .9]

\foreach \x/\y in {1/0,2/0,3/0, 4/0, 5/0,1/1,2/1,3/1, 4/1, 5/1,6/0,7/0,8/0, 9/0, 10/0,6/1,7/1,8/1, 9/1, 10/1}
 \draw[thick] (\x, \y) rectangle (\x+1, \y+1);

\draw[red, thick, ->] (1.55, 1.7) ..controls (4, 3.3) and (4.4, 1.3) .. (4.5, 0.75);

\draw (1.5, 2.8) node {$(x,\hat 9)$};
\draw (2.5, 2.8) node {$(x,\hat 8)$};
\draw (3.5 , 2.8) node {$(x,\hat 7)$};
\draw (4.5, 2.8) node {$(x,\hat 6)$};
\draw (5.5, 2.8) node {$(x,\hat 5)$};
\draw (6.5, 2.8) node {$(x,\hat 4)$};
\draw (7.5, 2.8) node {$(x,\hat 3)$};
\draw (8.5 , 2.8) node {$(x,\hat 2)$};
\draw (9.5, 2.8) node {$(x,\hat 1)$};
\draw (10.5, 2.8) node {$(x,\hat 0)$};

\draw (0 , 1.5) node {$x$ even};

\draw (1.5, 0.5) node {$+5$};
\draw (2.5, 0.5) node {$+5$};
\draw (3.5 , 0.5) node {$+3$};
\draw (4.5, 0.5) node {$+3$};
\draw (5.5, 0.5) node {$+1$};
\draw (6.5, 0.5) node {$-1$};
\draw (7.5, 0.5) node {$-1$};
\draw (8.5 , 0.5) node {$-3$};
\draw (9.5, 0.5) node {$-3$};
\draw (10.5, 0.5) node {$-5$};

\draw (0 , 0.5) node {$x$ odd};

\draw (1.5, 1.5) node {$+6$};
\draw (2.5, 1.5) node {$+4$};
\draw (3.5 , 1.5) node {$+4$};
\draw (4.5, 1.5) node {$+2$};
\draw (5.5, 1.5) node {$+2$};
\draw (6.5, 1.5) node {$0$};
\draw (7.5, 1.5) node {$-2$};
\draw (8.5 , 1.5) node {$-2$};
\draw (9.5, 1.5) node {$-4$};
\draw (10.5, 1.5) node {$-4$};

\end{tikzpicture}
\caption{The value table for  $\tb = 9$ and sufficiently large heap sizes $x$, even and odd respectively. The arrow indicates a Left win by the bid $\ell=3$, from  $(x,\hat 9)$, on a large even heap size $x$. This bid is not in equilibrium. In fact, the bid is dominated by Left bidding 1, which is in equilibrium, since $5+1=6$. Note also that bidding zero is in equilibrium. Indeed 0-ties are in equilibrium for all sufficiently large heap sizes.}
\label{fig:odd}

\end{center}
\end{figure}

\subsection{Property $\U$ for Unitary Games}\label{sec:propunitary}



We begin by proving monotonicity of maximin values for fixed heap sizes, i.e. property \U (A), `budget monotonicity' in Lemma~\ref{lem:col_mon}.

\begin{lemma} [Budget Monotonicity, \U (A)]\thlabel{lem:col_mon}
For all games, $\nym{x}{p}\ge \nym{x}{\pi}$ if $p\ge\pi$. 
\end{lemma}
\begin{proof}
The proof is by induction over the heap sizes, and note that $\nu_p(0,1)=0\ge 0=\nu_\pi(0,1)$. Suppose that the statement holds for all heap sizes smaller than $x>0$. In row $x$, Left has all the bidding options with budget $p$ as she has with budget $\pi$. Since, by induction, the values are non-decreasing in row $x-1$, if she can force a win (perhaps using a tie) with budget $\pi $, she is assured at least as good value in column $p$ as in column $\pi $. If she cannot force a win with budget $\pi $, the only difference is that with budget $p$ she could perhaps force a win, and again, the budget  monotonicity of the smaller heaps imply the result. 
\end{proof}
\begin{rem}\label{rem:col_mon}
By combining \eqref{eq:0-sum} with Lemma~ \ref{lem:col_mon}, we get  $\nu_p(x)\ge \nu_\pi(x)$ if $p\ge \pi$.
\end{rem}
Thus, by combining Lemma~\ref{lem:col_mon} and Remark~\ref{rem:col_mon}, we show that unitary games satisfies \thref{def:U} (A) property. However, we have some direct consequences of budget monotonicity. We may  condition maximin values on given bids. If the players tie $\ell$ from position $(x,\hat p)$, we write $\nym{x}{p}|_{T(\ell)}=-\nym{x-1}{q+\ell}+1$, and so on, for the result, given an $\ell$-tie at $\posm{x}{p}$ and subsequent maximin play.

In Lemma~\ref{lem:tie} prove the {\em tie monotonicity} property of for unitary games which shows that Left weakly prefers the tie $(\hat {\ell-1},\ell-1)$ over $(\hat \ell,\ell)$. This property will be useful to prove that unitary games satisfies $\U(B)$ property.

\begin{lemma}[Tie Monotonicity]\thlabel{lem:tie}
At any position, with $x,p,\ell>0$, the tie $(\hat{\ell},\ell)$ is maximin weakly worse for player Left,  than the tie $(\widehat{\ell-1},\ell-1)$. 
\end{lemma}
\begin{proof}
The proof follows by heap monotonicity  (\thref{lem:col_mon}). In particular, we have  $\nym{x}{q+\ell}\ge \nym{x}{q+\ell-1}$. On the other hand $\nym{x}{p}|_{T(\ell)}=1-\nym{x-1}{q+\ell}$ and $\nym{x}{p}|_{T(\ell-1)}=1-\nym{x-1}{q+\ell-1}$. Hence, $\nym{x}{p}|_{T(\ell)}\le\nym{x}{p}|_{T(\ell-1)}$.
\end{proof}
This innocent consequence of heap monotonicity implies that Right can assure the outcome of the $0$-tie, namely all entries in the upper right area weakly bounded by the main diagonal (in the game value matrix described in Example~\ref{ex:tb9}), are weakly smaller than the value at the 0-tie. (This property is implicit in the proof of \thref{thm:unique}.)

In the next Lemma~\ref{lem:marker}, we prove a natural bounds on game value of the equilibrium outcomes when Left has marker in term of the game value at equilibrium when Left does not have a marker. This result precisely prove that unitary games satisfies Definition~\ref{def:U} (B). As a bonus, we get a bound on how good the marker can be. 
\begin{lemma}[Marker Monotonicity, \U (B)]\thlabel{lem:marker}
Consider any unitary ruleset. Then, for all heap sizes $x$ and all Left budgets $p$, 
\begin{align}\label{eq:tie}
\ny{x}{p}\le \nym{x}{p}\le \ny{x}{p}+2.
\end{align}
\end{lemma}
\begin{proof}
The proof is by induction on the inequalities~\eqref{eq:tie} (and they clearly hold for $x=0$). 
Suppose that there is a position such that  $\ny{x}{p}>\nym{x}{p}$, i.e. Left does not want to win a tie at $\posm{x}{p}$. By \thref{lem:tie}, we assume that Left bids 0. Then, if Right accepts the 0-tie, induction gives a contradiction, namely $$1+\ny{x-1}{p}<\nym{x-1}{p}-1,$$
which is equivalent with $$2+\ny{x-1}{p}<\nym{x-1}{p}.$$
Suppose next that $2+\ny{x}{p}<\nym{x}{p}$, which means that a tie-win is worth more than 2 points. Suppose that Left bids 0. Then $2+\nym{x-1}{p}-1<\ny{x-1}{p}+1$, which contradicts the first inequality in \eqref{eq:tie}. 
\end{proof}

The next lemma is a consequence of Marker Monocity.

\begin{lemma}\thlabel{lem:markdom}
For all heap sizes $x$, all Left budgets $p$ and all Left bids $0\le\ell\le p$, $\nym{x}{p}\ge-\nym{x}{q+\ell}$. 
\end{lemma}
\begin{proof}
We have that, for all $x$, and all $p$, 
\begin{align*}
\hat\nu(x,p)&\ge \nu(x,p)=-\nym{x}{q}\ge -\nym{x}{q+\ell},
\end{align*}
by \thref{lem:marker}, because \thref{lem:col_mon} gives $\nym{x}{q+\ell}\ge\nym{x}{q}$. 
\end{proof}

In the next Lemma~\ref{lem:markworth}, we show that unitary games (rules) satisfies $\U(C)$. That is, the marker is never worth more than \$1. This will complete our argument for the proof of the main Theorem~\ref{thm:unitary}

\begin{lemma}[Marker Worth, \U (C)]\thlabel{lem:markworth}
Consider a unitary ruleset. Then, for all heap sizes $x$ and all Left budgets $p$,
\begin{equation*}
    \nym{x}{p}\leq \ny{x}{p+1}
\end{equation*}
\end{lemma}
\begin{proof}
We prove this theorem by an induction argument on heaps of size $x$, and analyze the values for all budget partitions. Consider the base case when $x = 1$. If Left can force a win with budget $p$ and the marker, then $p\ge q$, and hence $p+1>q-1$ gives that $\nym{1}{p} = 1$ implies $\ny{1}{p+1} = 1$. Therefore $\nym{1}{p}\leq \ny{1}{p+1}$.

Suppose next that the theorem is true for all heap sizes smaller than $x$. To complete the induction, we need to show that $\nym{x}{p}\leq \ny{x}{p+1}$.  We will prove the induction step for three different cases of maximin $\nym{x}{p}$:
\begin{itemize}
    \item[]\textbf{Case 1.} Left wins with tie at $\posm{x}{p}$  by bidding $\ell$;
    \item[]\textbf{Case 2.} Left wins without tie at $\posm{x}{p}$ by bidding $\ell$;
    \item[]\textbf{Case 3.} Right wins  at $\posm{x}{p}$ by bidding $r$.
\end{itemize}

 Consider {\bf Case 1}. Right maximin weakly prefers losing an item at tie $\ell$ rather than winning it by bidding $\ell+1$. Therefore, 
\begin{equation}
    -1 + \nym{x-1}{p+\ell+1} \geq 1+ \ny{x-1}{p-\ell} \label{eqn:mqcase1}
\end{equation}
There are two subcases. 

The first subcase is whenever $0\leq \ell \leq q-1$. We claim that Right does not maximin prefer winning if Left bids $\ell+1$ at position $\pos{x}{p+1}$. 
If Right wins at $\pos{x}{p+1}$ by bidding $\ell+1$ with marker, then
\begin{equation*}
    -1 + \nym{x-1}{p +\ell +2} \geq -1 + \nym{x-1}{p+\ell+1} \geq 1 + \ny{x-1}{p-\ell} 
\end{equation*}
The first inequality holds by heap monotonicity and the second follows from inequality~\eqref{eqn:mqcase1}. 

Now assume that Right bids $r> \ell+1$ at $\pos{x}{p+1}$, and wins. Then,  
\begin{equation*}
    -1 + \ny{x-1}{p +r+1} \geq -1 + \ny{x-1}{p+\ell+2} \geq -1 + \nym{x-1}{p+\ell+1} \geq 1 + \ny{x-1}{p-\ell}
\end{equation*}
The first inequality holds by heap monotonicity, the second by induction and the third follows from inequality~\eqref{eqn:mqcase1}. 

Altogether this implies that, for this subcase, $\ny{x}{p+1} \geq 1 + \ny{x-1}{p - \ell} = \nym{x}{p}$. 

Consider the other subcase, when $\ell \ge q$. In this case, Left can secure a win by bidding $\ell$ at $\pos{x}{p+1}$, as Right cannot over bid. This implies that $$\ny{x}{p+1} \geq 1 + \ny{x-1}{p - \ell +1} \geq 1 + \ny{x-1}{p - \ell}  = \nym{x}{p}$$ (by heap monotonicity). This completes the proof for \textbf{Case 1}.\\

We now prove \textbf{Case 2}. Left wins the bid, without a tie, by bidding $\ell $, and therefore, at $\posm{x}{p}$, Right maximin weakly prefers losing the bid. We get 
\begin{equation}
  -1 + \nym{x-1}{p+\ell +1} \geq 1 + \nym{x-1}{p-\ell} =  \nym{x}{p}  \label{eqn:mqcase2}
\end{equation}
Consider first the subcase, $\ell \leq q -1$. We claim that, if Left bids $\ell $, then she maximin wins at position $\pos{x}{p+1}$. If not, then either Right maximin wins by bidding $\ell$ with a marker or by bidding $r>\ell $. Now, if Right bids $\ell $ and wins using a marker, we study the inequalities
\begin{equation*}
    -1 + \nym{x-1}{p+\ell+1} \geq 1 + \nym{x-1}{p-\ell} \geq 1 + \ny{x-1}{p-\ell}
\end{equation*}
The first inequality holds by~\eqref{eqn:mqcase2} and the second inequality follows from marker monotonicity. Hence, at $\ny{x}{p+1}$, Right would prefer losing if Left bids $\ell$ over winning by bidding $\ell$ with marker. 

Now, if  Right bids $r> \ell$ and maximin wins the bid, we study the inequalities
\begin{align*}
    -1 + \ny{x-1}{p +r +1} &\geq -1 + \nym{x-1}{p+\ell+1}\\
    &\geq 1 + \nym{x-1}{p-\ell}\\
    &\geq  1 + \ny{x-1}{p-\ell}
\end{align*}
The first inequality holds by induction, the second follows from inequality~\eqref{eqn:mqcase2}, and the third follows from marker monotonicity. Hence, at $\ny{x}{p+1}$, Right would prefer losing if Left bids $\ell$ over winning by bidding $r>\ell$. 

The subcase $\ell\ge q$ is similar to {\bf Case 1}. Altogether,
\begin{equation*}
    \ny{x}{p+1} \geq 1 + \ny{x-1}{p - \ell +1} \geq 1 + \nym{x-1}{p - \ell} = \nym{x}{p},
\end{equation*}
which completes the proof of \textbf{Case 2}.\\

We now prove \textbf{Case 3}. As Right maximin wins at $\posm{x}{p}$ by bidding $r$, we get
\begin{equation}
   \nym{x}{p} = \nym{x-1}{p+r} -1 \leq \ny{x-1}{p-r+1} + 1 \label{eqn:mqcase3}
\end{equation}
In other words,  at $\posm{x}{p}$, Right prefers winning, by bidding $r$, over losing for marker, by bidding $r-1$. 

In order to prove this case, we claim that if Left bids $\ell=r-1$ at $(x,p+1)$, then Right prefers winning by a tie. This follows by the inequalities 
\begin{equation*}
    1 + \ny{x-1}{p-r+2} \geq 1 + \ny{x-1}{p-r+1} \geq \nym{x-1}{p+r} - 1 
\end{equation*}
The first inequality holds by budget monotonicity and the second follows from inequality~\eqref{eqn:mqcase3}. 
Hence, playing from $(x,p+1)$, Left can force at least the value where Right ties, i.e.  $\ny{x}{p+1}\geq \nym{x-1}{p+r} - 1 = \nym{x}{p}$. This completes the proof of \textbf{Case 3}.\\

Thus, the lemma holds.
%
%
%
\end{proof}

%
%
%

We have established the three properties of Definition~\ref{def:U}. Therefore we conclude with the proof of \thref{thm:unitary}.

\begin{proof}[Proof of \thref{thm:unitary}]
Combine \thref{lem:col_mon} with \thref{lem:markdom} and \thref{lem:markworth}. Thus Definition~\ref{def:U} applies: unitary games are in $\U$, and we may apply \thref{cor:U}.
\end{proof}

From now onward we use outcome, i.e. $o$-notation,  instead of $\hat \nu$, and in proofs, we may fix the bid of either player to prove inequalities as required by a given context. By fixing the bid of player Left (Right), we may obtain a lower (upper) bound of the outcome, corresponding to maximin (minimax) evaluation. 

\subsection{A Sign Border and Bounded Game Values}\label{sec:bounded}
In the last subsection, we have already shown that, for each position, unitary games exhibit a unique equilibrium. Now, in this subsection, we discuss the bounds on game values for given $\tb$, $x$ and budget partition. First, we establish the sign border of the game value for a given total budget of $\tb$. In Lemma~\ref{lem:signborder}, we show that if Left has more money than Right, i.e.$p \geq \ceil{\tb/2}$ then the game value $o_p(x)\geq 0$. Moreover, we show that the sign of the outcome is sensitive to the sign border $\ceil{\tb/2}$

\begin{lemma}[Sign Border]\thlabel{lem:signborder} 
For a given $\tb$ and all heap sizes $x$, $2p\ge\tb$ implies $\o{x}{p}\ge 0$, and  $2p<\tb$ implies $\o{x}{p}\le 0$. 
\end{lemma}
\begin{proof}
We prove that, if $p\ge \ceil{\tb/2}$ then $\o{x}{p}\ge 0$. If $p < \ceil{\tb/2}$ then $\o{x}{p}\le 0$.%

By the 0-sum property \eqref{eq:0-sum} combined with the first inequality in  Lemma~\ref{lem:marker}, we have $\o{x}{p}\ge -\o{x}{q}$. Adding $\o{x}{p}$ on both sides gives,
\begin{equation*}
    2\o{x}{p} \ge  \o{x}{p} - \o{x}{q}
\end{equation*}

As $ p \geq q $, using heap monotonicity lemma,  $2\o{x}{p} \ge  0$, which proves $\o{x}{p} \ge  0$.

For the second part when $p < \tb/2$, we use the second inequality in  Lemma~\ref{lem:marker} combined with the 0-sum property \eqref{eq:0-sum}, we have $\o{x}{p} \leq -\o{x}{q} + 2$.  Adding $\o{x}{p}$ on both sides, we get

\begin{equation*}
    2\o{x}{p} \leq  \o{x}{p} - \o{x}{q} +2.
\end{equation*}

Hence, 
\begin{equation*}
    \o{x}{p} \leq  (\o{x}{p} - \o{x}{q})/2 +1.
\end{equation*}
If $x$ is even, this implies that $\o{x}{p}\le 0$, by heap monotonicity and  Lemma~\ref{lem:pari}; if $x$ is odd then $\o{x}{p}\le 1$.

Thus we may assume that, for some $p<\tb/2$, $\o{x}{p}=1$. If $\o{x-1}{p-1}=0$, then if Left wins by bidding 1, Right increases to either get a tie by bidding 1 or to win by bidding 2. More precisely, if $\o{x-1}{\tb/2+2}=0$ then he 1-ties, and otherwise, he bids 2. In either case this gives $\o{x}{p}=-1$. 

If Left 0-ties at $(x,p)$ and $\o{x-1}{\lceil\tb/2\rceil}=0$, then Right bids 1 and obtains $\o{x}{p}\le -1$, with equality if $p= \lfloor\tb/2\rfloor$.
\end{proof}







\begin{obs}
Note that if the heap size $x$ is odd, then, by \thref{lem:pari}, the inequalities in \thref{lem:signborder} are strict.
\end{obs}

The previous Lemma~\ref{lem:signborder} is extensively useful to obtain an upper and lower bound on the game value. In the next Lemma~\ref{lem:bound}, by adapting \thref{lem:signborder}, we prove an upper and a lower bound of the game values.

\begin{lemma}[Bounded Outcome]\thlabel{lem:bound}
For all $x$ and all $p$, $-\floor{\tb/2}\le \o{x}{p}\le \ceil{\tb/2}+1$.
\end{lemma}
\begin{proof}
Consider first $p\ge \ceil{\tb/2}$. We define a Right strategy, such that Left cannot get better than $\ceil{\tb/2}+1$. Right bids 0 until (possibly) the first time Left's budget partition becomes smaller than $ \ceil{\tb/2}$. In the case where Left has the marker and bids 0, the budget partition will stay the same, but Right gets the marker. Hence, unless Left bids 1 in the next bid, the change in score would be $1-1=0$. When, at some point, Left bids 1, she keeps the marker, and gets a point, but the budget partition is one unit closer to the sign border.  Since the sign border is at $\ceil{\tb/2}$, the result follows, by the similar bound for the case $p < \ceil{\tb/2}$. 
\end{proof}

The following result in Proposition~\ref{lem:4th_ineq} proves the relation between change in the game value with added extra one dollar to Left's budget. We prove this result, even though we do not explicitly use it in order to prove other results. We wonder how this result generalizes for arbitrary subtraction sets.
\begin{prop}\label{lem:4th_ineq}
For all $\tb$, for all $x$, for all $p$, $\o{x}{p}+2\ge\o{x}{p+1}$.
\end{prop}
\begin{proof}
Assume for a contradiction that Left can obtain 
\begin{align}\label{eq:4gap}
\o{x}{p+1} = \o{x}{p}+4
\end{align}
(or more), where the outcome at $\posm{x}{p}$ is given. 
 If Right wins at $\posm{x}{p+1}$, then he will decrease bid, to possibly a tie. But Left prefers a smaller tie, so she can decrease and perhaps Right wins again. If so, if the gain is only 2 points we are done, so assume gain is 4 points. Then he decreased again and Left still prefers smaller tie, and so on. 
 
 By \thref{thm:unique} we do not need to consider cases where Left wins. Hence those cases where Right wins at $\posm{x}{p+1}$ reduce to study of cases where Left wins by a 0-tie at position $\posm{x}{p+1}$ and with \eqref{eq:4gap}. Thus, only two cases remain. 

\begin{enumerate}

     \item Left wins an $
     \ell$-tie at $\posm{x}{p}$,  and she wins a 0-tie at $\posm{x}{p+1}$.  
     \item Right wins at $\posm{x}{p}$ by bidding $r$, and Left wins a 0-tie at $\posm{x}{p+1}$.
\end{enumerate}

\noindent For {\bf Case 1}, we have $\o{x}{p}=-\o{x-1}{q+\ell}+1$ and $\o{x}{p+1} = -\o{x-1}{q-1}+1$, and hence, by \eqref{eq:4gap}
\begin{align}\label{eq:4gap2}
    4-\o{x-1}{q+\ell}=-\o{x-1}{q-1}
\end{align}
Thus, by induction, $\ell > 0$. This means that Left must avoid a Right win with $1\le r=\ell-1$ at $\posm{x}{p+1}$. Hence $\o{x-1}{p+r} <  -\o{x-1}{q+\ell}+1=-\o{x-1}{q-1}-3\le\o{x-1}{p+1}-3$, by \eqref{eq:4gap2} and marker monotonicity.  

For {\bf Case 2}, we have $\o{x}{p}=\o{x-1}{p+r}-1$ and $\o{x}{p+1} = -\o{x-1}{q-1}+1$, and hence 
\begin{align}\label{eq:4gap3}
    4+\o{x-1}{p+r}-1 = -\o{x-1}{q-1}+1,
\end{align}
That is, $\o{x-1}{p+r} +\o{x-1}{q-1}=-2$. If $q-1$ is to the left of the sign border, then $\o{x-1}{p+r} -\o{x-1}{p+1}\le -2$, which is impossible, by budget monotonicity and since $r\ge 1$. If $q-1$ is to the right of the sign border, then $p+r$ is to the left, and we get $\o{x-1}{q-1} -\o{x-1}{q-r}\le -2$, which is impossible, by $r\ge 1$ and budget monotonicity. 
\end{proof}


Next, in Lemma~\ref{lem:rowsplitmon} we prove that the equilibrium outcomes are monotone (non-decreasing or non-increasing as reflected by the sign border) for heap sizes of the same parity. We show that whenever $2p \ge \tb$, then the game value increases as heap size increases and similarly, for $2p < \tb$, game value is monotonically decreasing. We note that budget monotonicity does not imply this result as we do not change the parity but heap size.    

\begin{restatable}[Heap Monotonicity]{lemma}{lemmaheapmonotonisity}\thlabel{lem:rowsplitmon}
Fix a total budget $\tb$ and a Left budget $\hat p$, and consider unitary games with heap sizes of the same parity. If $2p\ge \tb$, then the value is monotonically non-decreasing, and otherwise it is monotonically non-increasing. This is impossible, by budget monotonicity, since $r\ge 1$.
\end{restatable}

\begin{proof}
The base case concerns heap sizes $x=0$ and $2$. Since the outcomes are all 0 at $x=0$, this case is covered by Lemma~\ref{lem:signborder} (Sign Border). 

We study the outcome at heap size $x$. The value $\o{x-2}{p}$ satisfies (at least) one out of three definitions:
\begin{itemize}
    \item[L:] $\o{x-2}{p}=\o{x-3}{p-\ell}+1$.\footnote{Note that by the proof of \thref{thm:unique}, it is not required to study the cases where Lefts wins a bid strictly. However, for the flow of the  proof we find it somewhat nicer when those cases are included.} 
    \item[R:] $\o{x-2}{p}=\o{x-3}{p+r}-1$.
    \item[T:] $\o{x-2}{p}=-\o{x-3}{q+\ell}+1$.
\end{itemize}
In case L, Left wins the bid; in case R, Right wins the bid, and in case T, there is a tie, so Left, who has the marker, wins the bid.

The proof splits into 10 distinct cases, depending on how the bid is won,
\begin{enumerate}
    \item Tie and $2p \ge \tb, q+\ell < \ceil{\tb/2}$
    \item Tie and $2p < \tb, q+\ell \ge \ceil{\tb/2}$
    \item Right wins and $2p\ge \tb, 2(p+r)\ge\tb$
    \item Left wins and $2p\ge \tb, 2(p-\ell)\ge \tb$
    \item Left wins and $2p < \tb, 2(p-\ell)<\tb$
    \item Right wins and $2p< \tb, 2(p+r)<\tb$
\end{enumerate}

\begin{enumerate}
    \item[A.] Left wins and $2p \ge \tb, 2(p-\ell) < \tb$
    \item[B.] Right wins and $2p < \tb, 2(p+r) \ge \tb$
    \item[C.] Tie and $2p\ge \tb, q+\ell\ge\ceil{\tb/2}$
    \item[D.] Tie and $2p< \tb, q+\ell< \ceil{\tb/2}$
\end{enumerate}
The case (D) cannot happen, because $q\ge \ceil{\tb/2}$ and $\ell\ge 0$. 

In cases 1-6 and several subcases we use induction. Suppose for example that $2p\ge \tb$, and we wish to prove, by induction, that for all $x>1$, 
\begin{align}\label{eq:ineq_ge}
\o{x}{p}\ge \o{x-2}{p}. 
\end{align}
Hence, we show that Right cannot do better in $(x,p)$ than in $(x-2,p)$. As induction hypothesis, we may assume $\o{x-1}{p+r}\ge \o{x-3}{p+r}$, and hence, if Right wins at $x$ by bidding $r$, then inequality \eqref{eq:ineq_ge} holds. We will refer to similar situations by saying `by induction'. For some more detail, to contradict the inequality \eqref{eq:ineq_ge}, Right must change strategy at $x$. That is he must lower his bid to $r-1$ or smaller. If he lowers to $r-1$, then we may assume that this is now a tie, and this situation has to be considered. In case the decrease of bid is successful for Right, then Left might deviate, etc;  the particular context will determine.

In cases of tie, the argument will be by induction in cases where the relevant budget partition crosses the Sign Border. For example in case (1), Left wins, and $\tb -p+\ell<\tb/2$. Since $2p\ge \tb$, we wish to prove a non-decreasing outcome, and induction thus applies when signs for non-decreasing outcomes change. Note that in case of tie bids, the player who tries to contradict the inequality, could do this either by lowering, or raising the bid, and thus several sub-cases may need to be considered. Since there are significant variations to why the contradicting player will not succeed, we will treat all cases.\\

\noindent {\bf Case 1.} The players $\ell$-tie at heap size $x-2$ and where $2p \ge \tb$. We have 
\begin{align}\label{eq:def_o}
\o{x-2}{p} = -\o{x-3}{q + \ell} + 1,
\end{align}
where $q + \ell < \ceil{\tb/2}$, i.e. $ 2(p-\ell)>\tb $. We need to prove that  
\begin{align}\label{eq:ineq_ge1}
\o{x}{p} \geq \o{x-2}{p}.    
\end{align}

We use induction to assume that
\begin{align}\label{eq:nondecr_tie}
     -\o{x-1}{q+\ell} \geq  -\o{x-3}{q+\ell}
\end{align}
   
By way of contradiction,  \eqref{eq:nondecr_tie} and \eqref{eq:ineq_ge1}, Right deviates at $x$: 
\begin{enumerate}[(1)]
    \item If Right decreases his bid, Left wins by bidding $\ell$. Hence, by $p-\ell \ge \tb/2$, induction gives $\o{x}{p}|_{L(\ell)} = 1+ \o{x-1}{p-\ell} \geq 1+ \o{x-3}{p-\ell}$.  
     By \thref{lem:tie}, Tie Monotonicity,  $\o{x-3}{p-\ell} \geq -\o{x-3}{q+\ell}$ and $1-\o{x-3}{q+\ell} = \o{x-2}{p}$.
    Thus, $\o{x}{p}|_{L(\ell)} \ge \o{x-2}{p}$.
    
    \item If Right increases his bid, he wins and we get $\o{x}{p}|_{R(r)} = \o{x-1}{p+r} - 1$, with $r=\ell + 1$. As Right does not benefit by increasing his bid at `$x-3$', he cannot benefit by increasing his bid at `$x-1$'. This follows by induction, because, by assumption \eqref{eq:def_o}, $\o{x-2}{p} \leq \o{x-3}{p+r} - 1 \leq
    \o{x-1}{p+r} - 1 = \o{x}{p}|_{R(r)}$.
\end{enumerate}

\noindent {\bf Case 2.} The players $\ell$-tie at heap size $x-2$ and where $2p < \tb$. We have 
\begin{align}\label{eq:def_o_2}
\o{x-2}{p} = -\o{x-3}{q + \ell} + 1,
\end{align}
where $q + \ell \ge \ceil{\tb/2}$, i.e. $ 2(p-\ell)\le\tb $. We need to prove that  
\begin{align}\label{eq:ineq_ge1_2}
\o{x}{p} \leq \o{x-2}{p}.    
\end{align}

We use induction to assume that
\begin{align}\label{eq:nondecr_tie_2}
     -\o{x-1}{q+\ell} \leq  -\o{x-3}{q+\ell}
\end{align}
By way of contradiction of  \eqref{eq:ineq_ge1_2}, Left deviates at $x$.

\begin{enumerate}[(1)]
    \item If Left decreases her bid, Right wins by bidding $\ell$. Hence, $\o{x}{p}|_{R(l)} = \o{x-1}{p+\ell} - 1$. If $2(p+\ell) < \tb $, by induction, Left cannot contradict \eqref{eq:ineq_ge1_2}. Since the relative loss for Left is 2, when Right wins the bid, by heap monotonicity, Left requires a relative 4-gap in outcomes at `$x-3$' and `$x-1$', with $2(p+\ell)\ge\tb$. That is, Left requires 
    \begin{align}\label{eq_4gap_c2}
    -\o{x-3}{q+\ell}+4\le \o{x-1}{p+\ell} 
    \end{align}
    By the assumption  $2p<\tb$, we get  $p+\ell<q+\ell$. Therefore, by heap monotonicity,  $\o{x-3}{q+\ell}\ge 2$. Moreover, if $\o{x-3}{q+\ell}= 2$, by heap monotonicity and \eqref{eq_4gap_c2}, this forces $0=\o{x-3}{p+\ell}<\o{x-1}{p+\ell}=2$. We will refer to this situation as \emph{the 4-gap principle}. 
    
    Hence, Right can decrease the bid to $\ell-1$, and still satisfy $\o{x-3}{q+\ell-1}= 2$, and the argument can be repeated with 
    $\ell-1$ instead of $\ell$, until at some point $\o{x-1}{p+\ell '}=0$; this must happen for some $\ell '\ge 0$, by $2p<\tb$ and Sign Border. The case $\o{x-3}{q+\ell}\ge 4$ need not be considered, since $2(p+\ell)\ge \tb$ and the Sign Border implies that Left should have decreased the bid to $\ell-1$ at `$x-3$' (she would have been strictly better off losing the bid and if Right decreases then by Tie Monotonicity, she benefits by a smaller tie). 
    
    \item If Left increases her bid, by induction, she cannot increase the outcome. Namely,  $\o{x}{p}|_{L(\ell+1)} = \o{x-1}{p-\ell-1} + 1\le \o{x-3}{p-\ell-1} + 1\le \o{x-3}{q+\ell}+1$, by assumption. 
\end{enumerate}

\noindent {\bf Case 3.} Right wins at heap size $x-2$, by bidding $r$, and where $2p \ge \tb$. Then $2(p+r) > \tb$. We have $\o{x-2}{p} = \o{x-3}{p+r} - 1 $, and need to prove that 
\begin{align}\label{eq:ineq_case3}
\o{x}{p} \geq \o{x-2}{p}. 
\end{align}
If Right keeps the same winning bid, by induction, he cannot contradict this inequality. Hence he decreases the bid to $r-1$ at $x$. We may assume this is a tie, so the relative loss for Right is 2 points. By induction, we may assume that $q+r-1\ge \ceil{\tb/2}$, and that  
\begin{align}\label{eq:jump_case3}
-\o{x-1}{q+r-1}<-\o{x-3}{q+r-1}
\end{align}
Suppose first that $\o{x-3}{p+r}-1=\o{x-2}{p}=0$. Then \eqref{eq:jump_case3} forces $q+r-1>p+r$. This contradicts the assumption $2p\ge \tb$.

Suppose next that 
$\o{x-3}{p+r}-1=\o{x-2}{p}=1$. Then 
\begin{align}\label{eq:2_case3}
2= \o{x-3}{p+r}\le \o{x-1}{p+r}, 
\end{align}
by induction.

{\it Claim}: $p+r>q+r-1$. {\it Proof of Claim}: If $p+r\le q+r-1$, then, by \eqref{eq:2_case3}, heap monotonicity and the 4-gap principle,  Right instead prefers to tie at `$x-3$'.

Therefore, by \eqref{eq:jump_case3}, \eqref{eq:2_case3}  and heap monotonicity, we get  $2=\o{x-1}{p+r-1}>\o{x-3}{p+r-1}=0$. 

But then, Left will decrease her bid at $x$, below $r-1$, and Right will win. However, the argument gives the same outcome as when he wins by bidding $r$. Thus, we may repeat the argument, and Right cannot contradict \eqref{eq:ineq_case3}.
\\

\noindent\textbf{Case 4.} 
Left wins by bidding $\ell>0$, $2p\ge \tb$ and $2(p-\ell)\ge \tb$. Thus, $\o{x-2}{p} = \o{x-3}{p-\ell} + 1 $. We need to prove that 
\begin{align}\label{eq:ineq_case4}
\o{x}{p} \geq \o{x-2}{p}. 
\end{align}
If Right keeps the same bid at $x$, by induction, he cannot contradict this inequality. If he can increase the bid to $\ell$ at $x$, there is a tie (otherwise we are done). If the tie remains to the left of the Sign Border, and there is an increase of outcome, i.e. $2(q+\ell)\ge\tb$ and $\o{x-1}{q+\ell} > \o{x-3}{q+\ell}$, then he will gain the sufficient amount to contradict \eqref{eq:ineq_case3}. But by the third assumption, this can only happen if $2(p-\ell)=\tb$. Hence $\tb$ is even. Hence,  $2=\o{x-1}{q+\ell} > \o{x-3}{q+\ell}
=0$, which would give outcome -1 at $x$ instead of +1. But, since the bids  $q+\ell=p-\ell$, Left can deviate and let Right win the bid, to produce an outcome at least +1, since $p+\ell>p-\ell$. \\

\noindent\textbf{Case 5.} 
Left wins by bidding $\ell>0$ and $2p < \tb$. We have $\o{x-2}{p} = \o{x-3}{p-\ell} + 1$.
We need to prove that 
\begin{align}\label{eq:ineq_case5}
\o{x}{p} \leq \o{x-2}{p}. 
\end{align}
By induction, 
\begin{align*}
    \o{x-1}{p-\ell} \leq \o{x-3}{p-\ell}.
\end{align*}
Therefore, to contradict \eqref{eq:ineq_case5}, Left must change her bid, and she can decrease to `$\ell -1$', to get a tie. But, 

\begin{align}
    \o{x}{p}|_{T(\ell-1)} &= -\o{x-1}{q+\ell-1} + 1 \le -\o{x-3}{q+\ell-1} + 1\\
    &\le \o{x-3}{p-\ell}+1 = \o{x-2}{p},
\end{align}

since $p < \ceil{\tb/2}$ implies that $q+\ell-1 \geq \ceil{\tb/2}$. If she tries to decrease further, to make Right win the bid to gain a higher outcome, then Right can respond by decreasing his bid, and the sequence of inequalities still holds.\\
 
\noindent \textbf{Case 6.} 
Right wins by bidding $r$ and $2(p+r) < \tb$. That is,  $\o{x-2}{p} = \o{x-3}{p+r} - 1<0$.  
We need to prove that 
\begin{align}\label{eq:ineq_case6}
\o{x}{p} \leq \o{x-2}{p}. 
\end{align}

And assume, by induction, 
\begin{align*}
     \o{x-1}{p+r} \leq  \o{x-3}{p+r}.
\end{align*}
Hence Left must change bid, to contradict \eqref{eq:ineq_case6}. 
If Left can increase her bid to `$r$' at `$x$', we get 
\begin{align*}
    \o{x}{p}|_{T(r)} &= -\o{x-1}{q+r} + 1\le -\o{x-3}{q+r} + 1\\
    &\le \o{x-3}{p+r}-1= \o{x-2}{p},
\end{align*}

since $p < \ceil{\tb/2}$ implies  that $q+r \geq \ceil{\tb/2}$.  If she can increase her bid to `$r+1$', since $2p<TB$, right can also raise his bid to $r+1$, and the sequence of inequalities still holds. \\

\noindent\textbf{Case A.} 
Left wins, by bidding $\ell$, $2p \ge \tb$ and  $2(p-\ell) < \tb$. We have 

\begin{align}\label{eqcaseA}
\o{x-2}{p} = \o{x-3}{p-\ell} + 1,
\end{align}

and we must prove that  
\begin{align}\label{tpcaseA}
\o{x}{p} \geq \o{x-2}{p}.    
\end{align}

By induction, 
\begin{align}\label{idcaseA}
     \o{x-1}{p-\ell} \leq  \o{x-3}{p-\ell}.
\end{align}

If the inequality is strict, then $-2\ge \o{x-1}{p-\ell}$, and  Left must change her bid. Suppose she decreases her bid to a `$0$' tie. Then 

\begin{align*}
    \o{x}{p}|_{T(0)} = -\o{x-1}{q} + 1\ge 1,
\end{align*}
by the assumption $p\ge \ceil{\tb/2}$, which suffices to justify \eqref{eqcaseA}. If Right instead wins by bidding 1, then this only contradicts \eqref{tpcaseA} if $\o{x-1}{p+1}=0$. In this case Left can raise the bid to get a 1-tie, and indeed, $\o{x-1}{q+1}\le \o{x-1}{p+1}=0$, by heap monotonicity, since $q\le p$; this implies $-\o{x-1}{q+1}\ge 0$. This argument can be repeated (Right instead wins by bidding 2 and perhaps Left raises to a 2-tie etc.) until one of the assumptions fails to hold.\\

\noindent\textbf{Case B.} 
Right wins, by bidding $r$, $2p < \tb$ and $2(p+r) \ge \tb$. Thus, 
\begin{align}\label{eqcaseB}
\o{x-2}{p} = -\o{x-3}{p+r} - 1\le -1,
\end{align}
and we need to prove that  
\begin{align}\label{tpcaseB}
\o{x}{p} \leq \o{x-2}{p}.    
\end{align}

By induction, 
\begin{align}\label{idcaseB}
     \o{x-1}{p+r} \geq  \o{x-3}{p+r}
\end{align}
If this is a strict inequality, then Right must change his bid to satisfy \eqref{tpcaseB}. He decreases to $r-1$, and gets a tie:
\begin{align*}
    \o{x}{p}|_{T(r-1)} = -\o{x-1}{q+r-1} + 1
\end{align*}
Note that $q+r-1\ge p+r$. Therefore, by assumption of strict inequality in \eqref{tpcaseB}, Right obtains the desired 4-gap, which  implies that the inequality \eqref{tpcaseB} holds. Thus 

 Left could decrease, and let Right win by $r-1$, but this could only help him, and we could either repeat the argument, or go to Case 6. Suppose that Left can increase to win by bidding $r$. But Right can `$r$'-tie, and this is weakly better for him than the assumed `$r-1$' tie.
\\

\noindent\textbf{Case C.} 
Tie, $2p\ge \tb$ and  $q+\ell\ge\ceil{\tb/2}$. We have 
\begin{align}\label{eq:eqcaseC}
\o{x-2}{p} = -\o{x-3}{q+\ell} + 1,
\end{align}

We need to prove that  
\begin{align}\label{eq:tpcaseC}
\o{x}{p} \geq \o{x-2}{p}.    
\end{align}

We use induction to assume that
\begin{align}\label{eq:idcaseC}
     -\o{x-1}{q+\ell} \leq  -\o{x-3}{q+\ell}
\end{align}

If the inequality is strict, then Left must change her bid. If she decreases so Right wins by bidding $\ell$, then 

\begin{align}
    \o{x}{p}|_{R(\ell)} = \o{x-1}{p+\ell} - 1.
\end{align}

By the assumption of strict inequality, together with  \thref{lem:markdom}, we get $2\le \o{x-1}{q+\ell}\le \o{x-1}{p+\ell}$, which gives the desired 4-gap. Then, if Right decreases, \thref{lem:tie} gives that Left cannot be worse of by tie `$\ell-1$'. The argument can be repeated, or we are in Case 1.

If the inequality \eqref{eq:idcaseC} is not strict, then Right must change his bid to contradict \eqref{eq:tpcaseC}. Induction shows that he cannot benefit by increasing his bid. Suppose he decreases so Left wins by bidding $\ell$. By \thref{lem:markdom}, $\o{x-1}{p-\ell}\ge -\o{x-1}{q+\ell}$, which does not worsen Left's result.
\end{proof}

We are now ready to prove the second main theorem of the paper. In Theorem~\ref{thm:const}, we show that the game value of a given budget partition is constant for large heapsizes. The proof of the Theorem follows from Lemma~\ref{lem:bound} and Lemma~\ref{lem:rowsplitmon}.

\begin{thm}[Eventual period 2 of equilibrium outcomes]\thlabel{thm:const}
The game value of a given budget partition is constant, for all sufficiently large heaps of the same parity. 
\end{thm}

\begin{proof}

Fix any parity for the heap sizes. By \thref{lem:rowsplitmon}, the game values are column-wise non-decreasing weakly to the left of $\tb/2$, and non-increasing to the right of $\tb/2$. Therefore, since, for each column, by \thref{lem:bound} their absolute values are bounded, they converge to a finite constant. 
\end{proof}
\subsection{A Bidding Automaton and a Quadratic Bound}\label{sec:automaton}
In this subsection, we analyze the convergence of the game value for unitary games. We will determine a quadratic bound in the total budget $\tb$ for the demonstrated game value `convergence'. We will do this via explicit bounds of the outcome vector. 

We make use of functions defined on the even and odd integers, that later will represent the possible budget partitions for even and odd total budgets, respectively. See \thref{lem:bidaut} below. Define nearest integer functions $\alpha_{\even}$ and $\alpha_{\odd}$ (the indexes will correspond to the parities of the heap sizes) on the {\bf even} integers (for \tb\ even inputs will correspond to  $p-q=2p-\tb$), by
\begin{equation*}
  \alpha_{\even}(\delta)=\begin{cases}
    \floor*{\frac{\delta+1}{2}}, & \text{if } \delta\equiv 0\pmod 4;\\\\
    \ceil*{\frac{\delta+1}{2}}, & \text{otherwise}.
  \end{cases}
\end{equation*}

\begin{equation*}
  \alpha_{\odd}(\delta)=\begin{cases}
    \ceil*{\frac{\delta+1}{2}}, & \text{if } \delta\equiv 0\pmod 4;\\\\
    \floor*{\frac{\delta+1}{2}}, & \text{otherwise}.
  \end{cases}
\end{equation*}

Let $\iota:\Z\rightarrow \{0,1\}$ be the function $\iota(x)=1$ if and only if $x>0$. Define a nearest integer function, $\beta$ on the {\bf odd}  integers (for \tb\ odd inputs will correspond to  $p-q=2p-\tb$),  by
\begin{equation*}
  \beta(\delta)=\begin{cases}
    \floor*{\frac{\delta}{2}}+\iota(\delta), & \text{if } \delta\equiv 1\pmod 4.\\\\
    \ceil*{\frac{\delta}{2}}+\iota(\delta), & \text{otherwise}.
  \end{cases}
\end{equation*}





For any fixed (budget) $\tb\in\N_0$, we define an automaton $\a$ with $\tb+1$ nodes and 2 states per node, such that, for all states $j\in\{\even,\odd\}$,  for all nodes $p\in\{0,\ldots , \tb\}$ $$(j,p)\overset{\a}{\longrightarrow} (j^c,q),$$  with updates, for all $j,p$,
$$\a(j,p)=1-\a(j^c,q),$$ 
where $j^c$ is the complement of $j$, 
and where initial values are assigned to say all even states. Note that, by definition, independently of initial values, the even states are reflexive, and so are the odd ones. In Lemma~\ref{lem:bidaut}, we show that the $\alpha$ and $\beta$  functions are automaton $\a$ duals in the following sense.

\begin{lemma}\thlabel{lem:bidaut}
For all $p$, let $\a(\evenheap,p) = \alpha_\evenheap ( 2p - \tb )$. Then, for all $p$, $\a(\oddheap,p) = \alpha_\oddheap( 2p - \tb )$. For all $p$, let $\a( \evenheap,p) = \beta ( 2p - \tb )$. Then, for all $p$, $\a(\oddheap,p) = \beta( 2p - \tb )$. 
\end{lemma}
\begin{proof}
In case of even \tb, we want to justify that $\alpha_\even(2p-\tb) = 1-\alpha_\odd(\tb-2p)$, which holds since, for all $\delta$, $\ceil{\frac{\delta +1}{2}}+\floor{\frac{-\delta+1}{2}}=1$. Namely, if $\delta$ is odd, then we cancel the nearest integer functions and the equality holds; if $\delta=2m$ is even, then we get $m+\ceil{1/2}-m+\floor{1/2}=1$.

In case of odd \tb, we want to justify that $\beta(2p-\tb) = 1-\beta(\tb-2p)$. Suppose first that $p > \tb/2$. Then $\beta(2p-\tb)+\beta(\tb-2p) = \floor{\delta/2}+1+\ceil{-\delta/2}$ or $\ceil{\delta/2}+1+\floor{-\delta/2}$; in either case these expressions equal one. Because \tb\ is odd, the other case is $i<\tb/2$, and so the argument is analogous. 
\end{proof}

A (generic) \emph{$\tb$-bidding automaton} is a finite state machine with `$\tb + 1$' states, directed edges between the states, and an update rule for each directed edge. Each state has an outgoing edge corresponding to a feasible winning bid from one of the players. Clearly, for any \tb,\ \a\ is a bidding automaton where the winning bid is 0, i.e.  Left wins a 0-tie. We improve on \thref{thm:const}, by describing an explicit bound, and begin with a lemma. 

\begin{lemma}\thlabel{lem:boundaut}
For any given total budget \tb, and any Left budget $\hat p$, the entries of automaton $\a$ bound the outcome $\o{\cdot}{p}$. If $p\ge \tb/2$, then $\o{x}{p}\le\a(j,p)$, where the parity of $x$ is $j$, and otherwise $\o{x}{p}\ge\a(j,p)$.
\end{lemma}
\begin{proof}
We prove, by induction that the values as prescribed by automaton $\a$ cannot be exceeded. 

Consider first even $\tb$. In this case we are concerned with the $\alpha$ functions. For odd heap sizes $x$, and $2p\ge TB$, we show that 
\begin{enumerate}
\item $\o{x}{p+r}-1\ge\alpha_\even(2p-\tb)$, $r>0$,
\item $\o{x}{p-\ell}+1\le\alpha_\even(2p-\tb)$, $\ell>0$, and 
\item $-\o{x}{q+\ell}+1\le\alpha_\even(2p-\tb)$, $\ell\ge 0$. 
\end{enumerate}

That is, by induction, we show,

\begin{enumerate}
\item $\alpha_\odd(2(p+r)-\tb)-1\ge\alpha_\even(2p-\tb)$, $r>0$,
\item $\alpha_\odd (2(p-\ell)-\tb)+1\le\alpha_\even(2p-\tb)$, $\ell>0$, and 
\item $-\alpha_\odd (2(q+\ell))-\tb)+1\le\alpha_\even(2p-\tb)$. 
\end{enumerate}

Note that Case 3 has already been justified for $\ell=0$ in \thref{lem:bidaut}, and when $\ell>0$ obviously the inequality still holds.

For Case 2, the tightest situation is when $\ell=1$ i.e. when $2(p-1)-\tb\equiv 2\pmod 4$, i.e. if the outcome equals $\floor*{\frac{2(p-1)-\tb+1}{2}}+1$ and we see that it equals $ \floor*{\frac{2p-\tb+1}{2}}$, i.e., this is the case when $2p-\tb\equiv 0\pmod 4$ and  $2(p-1)-\tb\equiv 2\pmod 4$. If $\ell = 2$ then the outcome equals $\ceil*{\frac{2(p-2)-\tb+1}{2}}+1=\ceil*{\frac{2p-\tb+1}{2}}-1\le \floor*{\frac{2p-\tb+1}{2}}$. If $\ell>2$ the inequality is immediate.

The remaining case for even \tb\ and the cases for odd \tb\ are justified analogously.
\end{proof}

We say that a game \emph{converges at heap size} $x$ if, for all $p$, then $\o{x}{p}=\o{x+2}{p}$, but there is a $p$ such that $\o{x}{p}\ne \o{x-2}{p}$. Observe that convergence at $x$ implies that $\o{x+1}{p}=\o{x+3}{p}$, for all $p$. We say that a game \emph{converges} if it converges at some $x<\infty$.

\begin{thm}[Convergence Bound]\thlabel{thm:convbou}
The upper bound for convergence is the sum of the entries in the 0-bidding automaton, and it is of order of magnitude $O(\tb^2)$.
\end{thm}
\begin{proof}
We use the bounds of the outcomes for each budget partition as prescribed by the 0-bidding automaton \a. By \thref{lem:boundaut} if the game did not converge before the entries of \a\ have been reached, it converges at the first occurrence of the \a\ entries. We use \thref{thm:budadv} to see that (nearly) half the entires change by going from heap size 0 to heap size 2. 

Thus, for an even total budget \tb, we bound the maximal number of rounds induced by the 0-bidding automaton \a, as $$1+\sum_{\delta\in\{0,2,\ldots,\tb\}} \left(\frac{\delta+1}{2}+\frac{\delta+1}{2}\right)-\tb= 1+(\tb/2+1)\tb/2-\tb/2=O(\tb^2),$$ and a similar convergence bound holds for the odd sized total budgets, 
$$1+\tb/2+\sum_{\delta\in\{1,3,\ldots,\tb\}} \left(\frac{\delta}{2}+\frac{\delta}{2}\right)-\tb = 1+\ceil*{\tb/2}^2-\tb/2=O(\tb^2),$$
\end{proof}

If the upper bounds are obtained, then this proves that the equilibrium bids are zero-ties. Namely, the edges of the automaton $\a$ correspond to 0-bids. (There may be other optimal bids as well, but we do not classify those here.) In the next section, we illustrate some feasible bids for automaton $\mathcal A$. 

\begin{conj}[Automaton Game Correspondence]\label{conj:conv}
Consider any unitary game. The entries of the corresponding automaton $\mathcal A$ are obtained as outcomes, for all heaps of size at least $O(\tb^2)$. 
\end{conj}

\subsection{The Number of Forced Wins}
As an independent result, we count the number of forced wins a player with a larger budget can have.

\begin{thm}[Budget Advantage]\thlabel{thm:budadv}
Consider any game $(\tb;x,p,m;c)$, with $\tb=p+q$. Suppose that  Left has the marker. Then Left can force a win of the $x$ final moves if $p \geq (2^x-1)q + 2^{x-1} - 1$. If Right has the marker, then Left can force a win of the $x$ final moves if $p \geq (2^x-1)(q+1)$.
\end{thm}

\begin{proof}
We take the case when we have the configuration of the game as game $(\tb;x,p,1;c)$, where $p+q=\tb$ and $p \geq q$. To win the first move Left should at least bid $q$ dollars, i.e. $p \geq q$. So, Right has now at least $2q$ dollars and so Left to win the second round he must bid $2q+1$ dollars. He must have $q + (2q+1) = 3q+1$ dollars to win $2$ consecutive moves. Right has now at least $4q+1$ dollars. Left must bid $4q+2$ dollars to win the third consecutive round and in total he must have at least $q + (2q+1)+(4q+2) = 7q+3$ dollars. Similarly, he should bid $8q+4$ and $16q+8$ dollars to win the fourth and fifth consecutive moves. In total, he must have at least $q + (2q+1) + (4q+2) + (8q+4) + (16q+8) = 31q+15$ dollars to win 5 consecutive moves. 

We prove this by induction. We take the base case of $x=1$ move and we get $p \geq q$ which is true since if the budget is equal, Left wins by the marker. We assume it to be true for $x=k$ moves and prove it for $x=k+1$ moves. 

To win $x=k$ consecutive moves, Left's budget should be at least $(2^k-1)q + 2^{k-1} - 1$. To win $(k+1)^{\rm st}$ move, Left should bid at least `1' more than that of Right budget after $k$ move which is $(2^k-1)q + 2^{k-1} - 1 + q$, which is equal to $2^kq + 2^{k-1} - 1$. Hence Left must bid $2^kq + 2^{k-1}$. Hence total budget which must be available with Left after $k+1$ move is $\{(2^k-1)q + 2^{k-1} - 1\} + \{2^kq + 2^{k-1}\}$ = $2$ = $(2^{k+1}-1)q + 2^k - 1$. Hence induction holds.

When we have the configuration of the game $(\tb;x,p,0;c)$, Left should  at least bid $q+1$ dollars. Right now has $2q+1$ dollars. So, Left must have at least $(q+1) + (2q+2) = 3q+3$ dollars to win $2$ consecutive moves. Right now has $4q+3$ dollars for the third round and so Left must have at least $(q+1) + (2q+2) + (4q+4) = 7q + 7$ dollars to win $3$ consecutive moves.

We prove this by induction. We take the base case of $x=1$ move and we get $ p \geq q+1 $ which is true since if budgets are equal, Left wins by bidding `1' more than that of the bid of Right, which can be at most $q$.  We assume it to be true for $x=k$ moves and prove it for $x=k+1$ moves. 

To win $x=k$ consecutive moves, Left's budget should be at least $(2^k-1)(q+1)$. To win $(k+1)^{st}$ move, Left should bid at least `1' more than that of Right budget after $k$ move which is $(2^k-1)(q+1) + q$, which is equal to $2^k(q+1) - 1$. Hence Left must bid   $2^k(q+1)$. Hence total budget which must be available with Left after $k+1$ move is $\{2^k(q+1)\} + \{(2^k-1)(q+1)\}$  = $(2^{k+1}-1)(q+1) $. Hence induction holds.
\end{proof}


\section{Discussion}
Note that one can deduce neat formulas for the upper bounds for the convergence of the outcomes, by using the explicit bounds of outcome values as prescribed by the $\alpha$ and $\beta$ functions. The remaining question is if this worst possible convergence bound is tight. We believe so, because of the elegance of the 0-bidding automaton \a. But currently, we do not have an argument to show that earlier convergence could not happen, for some large total budget. There are other open questions: 1) Prove or disprove periodicity of outcomes for any \tb, but with an arbitrary finite subtraction set. Our methods indicate that if we restrict the allowed bids of the two players, then we still have convergence, but how do restricted bidding sets affect the strategies? In particular, what happens if 0-bids are not allowed? Classify asymmetric (partizan) bidding sets according to (asymptotic) player strength. 

On another note, an interesting paper on `general sum' Richman games has recently appeared \cite{MKT2018}. They show budget monotonicity for games on binary trees, but find a counterexample if a node may have three children (a threat provokes a situation where a player prefers a smaller budget). Our setting readily generalizes to general sum, or more specifically to so-called self interest games, by instead of the zero-sum definition, letting both players maximize their individual final scores. The discrete Richman bidding scheme would stay the same, and one would need to expand on various questions of monotonicity, and for example existence of a unique PSPE and Pareto efficiency.

\begin{acks}
This work began with informal discussions at the (Chocolate Caf\'e, where you can order hot chocolate with your chocolate bar and not coffee, nearby the) Combinatorial Game Theory Colloquium III, Lisbon, 22-24 January, 2019, organized and hosted by Carlos P. dos Santos, Lisbon, Portugal. The main part of this work was done while Ravi Kant Rai was visiting National University of Singapore in Summer 2019 and while Urban Larsson was visiting Indian Institute of Technology Bombay,  Fall 2019. We are grateful to our hosts Dr. Yair Zick and Prof. K.S. Mallikarjuna Rao for many valuable suggestions and discussions.
\end{acks}
\section{Illustration of feasible, dominated and relevant bids}\label{sec:feasdom}\label{sec:illustation}
Consider (symmetric) BCS. We illustrate the feasible bids for total budget 5, including a short discussion of dominated bids. Such bids are also feasible arrows for  automaton $\mathcal A$. In Figure~\ref{fig:TTB5}, we show the possible tie bids (and here domination is never an issue). In Figure~\ref{fig:LWTB5}, we show the possible Left winning bids, and in Figure~\ref{fig:LWDTB5}, we show the same bids, but without dominated bids. Note that whenever property \U\ holds, then Left winning bids may be ignored, in optimal play. In Figure~\ref{fig:RWDTB5}, we illustrate the corresponding situations for Right winning bids, and here the pictures (where dominated bids have been erased) are relevant. By {\em relevant}, we mean that, in general, one cannot exclude the possibility that a bid may be a unique equilibrium (in some specific setting).


\begin{figure}[ht!]
\begin{center}
\begin{tikzpicture}[scale = 0.8]
\begin{scope}[every node/.style={circle,thick,draw}]  
    \node (0) at (6,0) {0};
    \node (1) at (3,0) {1};
    \node (2) at (0,0) {2};
    \node (3) at (-3,0) {3};
    \node (4) at (-6,0) {4};
    \node (5) at (-9,0) {5};
    \draw (0, 3.5) {};
\end{scope}

\begin{scope}[>={Stealth[black]},
              every node/.style={fill=white,circle},
              every edge/.style={draw=blue, thick}]
    \path [<->] (5) edge[bend left=30] node {0T} (0);
    \path [<->] (4) edge[bend left=30] node {0T} (1);
    \path [<->] (3) edge[bend left=30] node {0T} (2);
\end{scope}
\end{tikzpicture}

\begin{tikzpicture}[scale = 0.8]
\begin{scope}[every node/.style={circle,thick,draw}]  
    \node (0) at (6,0) {0};
    \node (1) at (3,0) {1};
    \node (2) at (0,0) {2};
    \node (3) at (-3,0) {3};
    \node (4) at (-6,0) {4};
    \node (5) at (-9,0) {5};
     \draw (0, 4) {};
\end{scope}

\begin{scope}[>={Stealth[black]},
              every node/.style={fill=white,circle},
              every edge/.style={draw=blue, thick}]
    \path [<-] (5) edge[bend left=40] node {1T} (1);
    \path [<->] (4) edge[bend left=50] node {1T} (2);
    \path[<->] (3) edge [out=120,in=60,distance=9mm,swap]   (3);
\end{scope}
\end{tikzpicture}

\begin{tikzpicture}[scale = 0.8]
\begin{scope}[every node/.style={circle,thick,draw}]  
    \node (0) at (6,0) {0};
    \node (1) at (3,0) {1};
    \node (2) at (0,0) {2};
    \node (3) at (-3,0) {3};
    \node (4) at (-6,0) {4};
    \node (5) at (-9,0) {5};
    \draw (0, 3.5) {};
\end{scope}

\begin{scope}[>={Stealth[black]},
              every node/.style={fill=white,circle},
              every edge/.style={draw=blue, thick}]
    \path [<-] (5) edge[bend left=40] node {2T} (2);
    \path [<-] (4) edge[bend left=50] node {2T} (3);
\end{scope}
\end{tikzpicture}

\caption{The pictures represent Left's wins via tie bids, for $\tb = 5$; Left has the marker together with the indicated number of dollars. Note that only in the case of both players bidding ``0'' all nodes have outgoing edges. In either case, they are all \emph{negative}, indicated with the color blue in the picture. For example, if Left has \$4, and wins by bidding 0, then the next state is that Right gets the marker and \$1. Of course none of the bids are dominated in the case of a win by using the marker. Dominated bids only appear because the other player does not have enough budget to motivate such a bid. Compare this situation with Figures~\ref{fig:LWTB5} and \ref{fig:LWDTB5}}\label{fig:TTB5}
\end{center}
\end{figure}
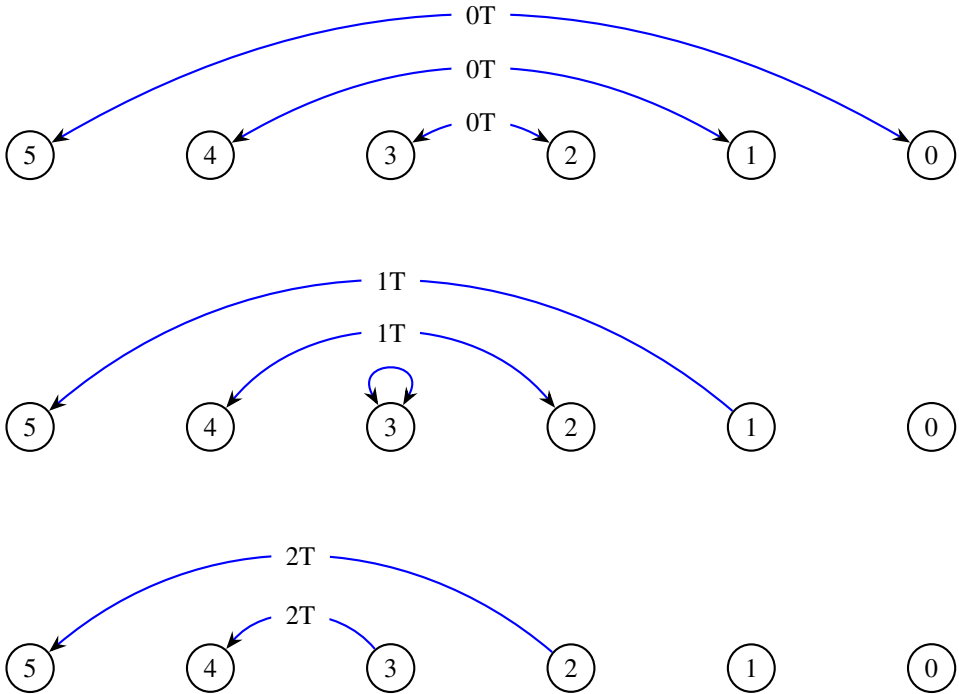

\begin{figure}
\begin{center}
\begin{tikzpicture}[scale = 0.8]

\begin{scope}[every node/.style={circle,thick,draw}]  
    \node (0) at (6,0) {0};
    \node (1) at (3,0) {1};
    \node (2) at (0,0) {2};
    \node (3) at (-3,0) {3};
    \node (4) at (-6,0) {4};
    \node (5) at (-9,0) {5};
\end{scope}
\begin{scope}[>={Stealth[black]},
              every node/.style={fill=white,circle},
              every edge/.style={draw=red,thick}]
    \path [<->] (5) edge[bend left=60] node {1W} (4);
    \path [<->] (4) edge[bend left=60] node {1W} (3);
    \path [<->] (3) edge[bend left=60] node {1W} (2);
    \path [<->] (2) edge[bend left=60] node {1W} (1);
    \path [<->] (1) edge[bend left=60] node {1W} (0);
\end{scope}
\end{tikzpicture}

\begin{tikzpicture}[scale = 0.8]
\begin{scope}[every node/.style={circle,thick,draw}]  
    \node (0) at (6,0) {0};
    \node (1) at (3,0) {1};
    \node (2) at (0,0) {2};
    \node (3) at (-3,0) {3};
    \node (4) at (-6,0) {4};
    \node (5) at (-9,0) {5};
    \draw (0, 3) {};
\end{scope}
\begin{scope}[>={Stealth[black]},
              every node/.style={fill=white,circle},
              every edge/.style={draw=red,thick}]
    \path [<->] (5) edge[bend left=45] node {2W} (3);
    \path [<->] (4) edge[bend left=45] node {2W} (2);
    \path [<->] (3) edge[bend left=45] node {2W} (1);
    \path [<->] (2) edge[bend left=45] node {2W} (0);
\end{scope}
\end{tikzpicture}

\begin{tikzpicture}[scale = 0.8]
\begin{scope}[every node/.style={circle,thick,draw}]  
    \node (0) at (6,0) {0};
    \node (1) at (3,0) {1};
    \node (2) at (0,0) {2};
    \node (3) at (-3,0) {3};
    \node (4) at (-6,0) {4};
    \node (5) at (-9,0) {5};
    \draw (0, 3) {};
\end{scope}
\begin{scope}[>={Stealth[black]},
              every node/.style={fill=white,circle},
              every edge/.style={draw=red,thick}]
    \path [<->] (5) edge[bend left=30] node {3W} (2);
    \path [<->] (4) edge[bend left=30] node {3W} (1);
    \path [<->] (3) edge[bend left=30] node {3W} (0);
\end{scope}
\end{tikzpicture}

\begin{tikzpicture}[scale = 0.8]
\begin{scope}[every node/.style={circle,thick,draw}]  
    \node (0) at (6,0) {0};
    \node (1) at (3,0) {1};
    \node (2) at (0,0) {2};
    \node (3) at (-3,0) {3};
    \node (4) at (-6,0) {4};
    \node (5) at (-9,0) {5};
     \draw (0, 3) {};
\end{scope}
\begin{scope}[>={Stealth[black]},
              every node/.style={fill=white,circle},
              every edge/.style={draw=red,thick}]
    \path [<->] (5) edge[bend left=30] node {4W} (1);
    \path [<->] (4) edge[bend left=30] node {4W} (0);
\end{scope}
\end{tikzpicture}

\begin{tikzpicture}[scale = 0.8]
\begin{scope}[every node/.style={circle,thick,draw}]  
    \node (0) at (6,0) {0};
    \node (1) at (3,0) {1};
    \node (2) at (0,0) {2};
    \node (3) at (-3,0) {3};
    \node (4) at (-6,0) {4};
    \node (5) at (-9,0) {5};
     \draw (0, 3) {};
\end{scope}
\begin{scope}[>={Stealth[black]},
              every node/.style={fill=white,circle},
              every edge/.style={draw=red,thick}]
    \path [<->] (5) edge[bend left=20] node {5W} (0);
\end{scope}
\end{tikzpicture}
\caption{Non-reduced bidding, for $\tb = 5$, where a right (left) pointing edge indicates that Left (Right) wins the bidding. As before the default is that Left has the marker together with the indicated number of dollars. Note the increasing number of nodes without outgoing edges.}\label{fig:LWTB5}
\end{center}
\end{figure}
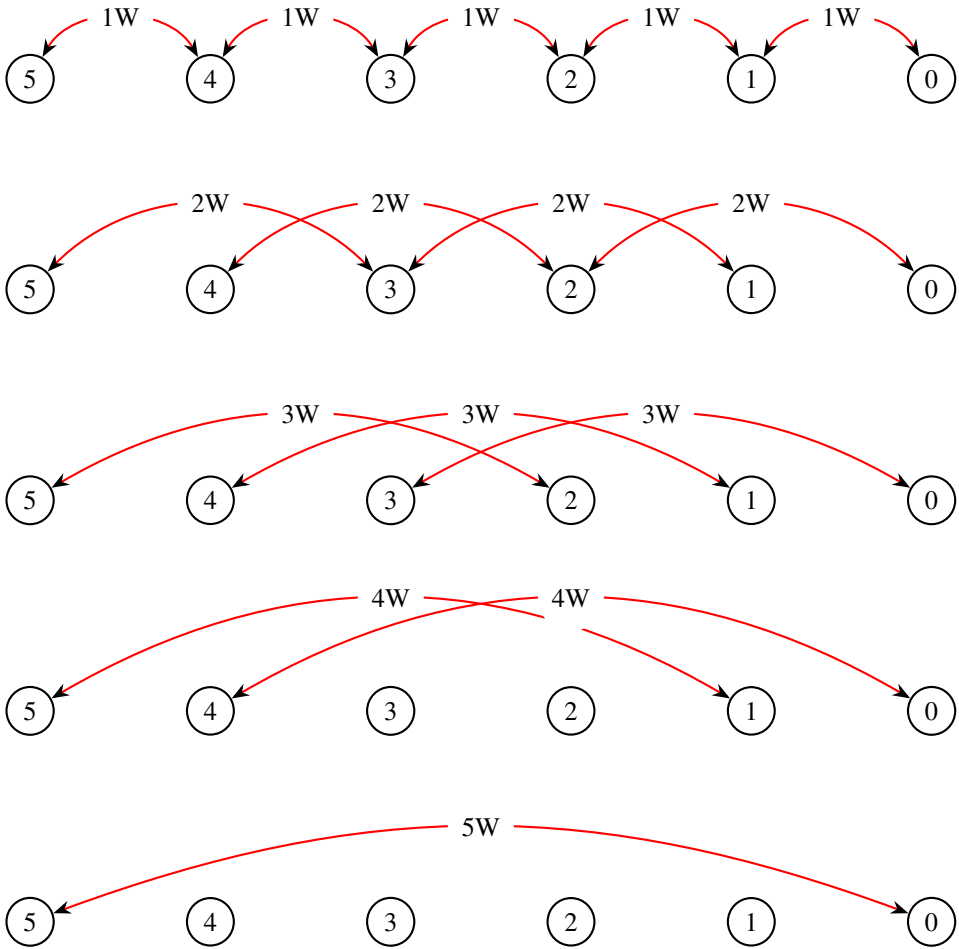


\begin{figure}
\begin{center}
\begin{tikzpicture}[scale = 0.8]

\begin{scope}[every node/.style={circle,thick,draw}]  
    \node (0) at (6,0) {0};
    \node (1) at (3,0) {1};
    \node (2) at (0,0) {2};
    \node (3) at (-3,0) {3};
    \node (4) at (-6,0) {4};
    \node (5) at (-9,0) {5};
\end{scope}
\begin{scope}[>={Stealth[black]},
              every node/.style={fill=white,circle},
              every edge/.style={draw=red,thick}]
    \path [->] (5) edge[bend left=60] node {1W} (4);
    \path [->] (4) edge[bend left=60] node {1W} (3);
    \path [->] (3) edge[bend left=60] node {1W} (2);
    \path [->] (2) edge[bend left=60] node {1W} (1);
    \path [->] (1) edge[bend left=60] node {1W} (0);
\end{scope}
\end{tikzpicture}

\begin{tikzpicture}[scale = 0.8]
\begin{scope}[every node/.style={circle,thick,draw}]  
    \node (0) at (6,0) {0};
    \node (1) at (3,0) {1};
    \node (2) at (0,0) {2};
    \node (3) at (-3,0) {3};
    \node (4) at (-6,0) {4};
    \node (5) at (-9,0) {5};
    \draw (0, 3) {};
\end{scope}
\begin{scope}[>={Stealth[black]},
              every node/.style={fill=white,circle},
              every edge/.style={draw=red,thick}]
    \path [->] (4) edge[bend left=45] node {2W} (2);
    \path [->] (3) edge[bend left=45] node {2W} (1);
    \path [->] (2) edge[bend left=45] node {2W} (0);
\end{scope}
\end{tikzpicture}

\begin{tikzpicture}[scale = 0.8]
\begin{scope}[every node/.style={circle,thick,draw}]  
    \node (0) at (6,0) {0};
    \node (1) at (3,0) {1};
    \node (2) at (0,0) {2};
    \node (3) at (-3,0) {3};
    \node (4) at (-6,0) {4};
    \node (5) at (-9,0) {5};
    \draw (0, 3) {};
\end{scope}
\begin{scope}[>={Stealth[black]},
              every node/.style={fill=white,circle},
              every edge/.style={draw=red,thick}]
    \path [->] (3) edge[bend left=30] node {3W} (0);
\end{scope}
\end{tikzpicture}

\caption{The remaining Left winning bids from Figure~\ref{fig:LWTB5} when dominated bids have been erased, for $\tb = 5$. These bids are not relevant whenever property $\U$ is satisfied.}\label{fig:LWDTB5}
\end{center}
\end{figure}
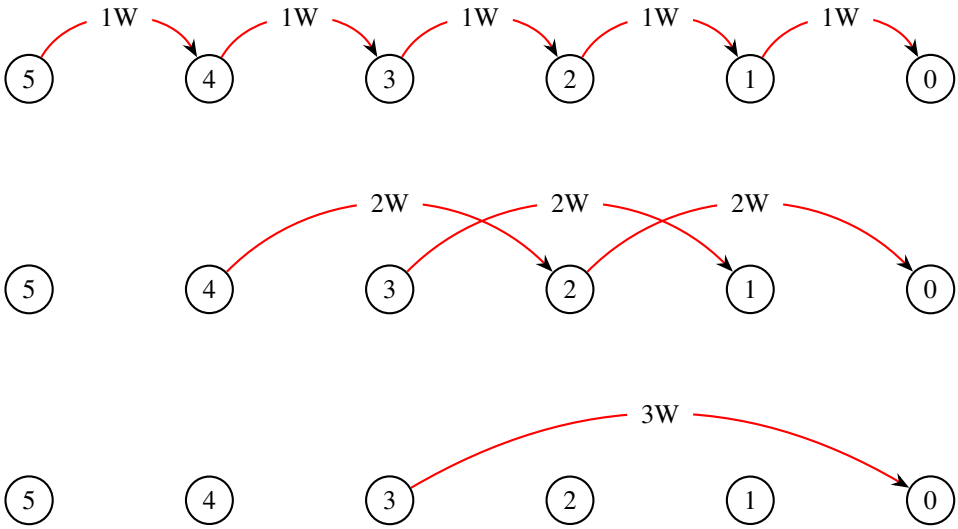

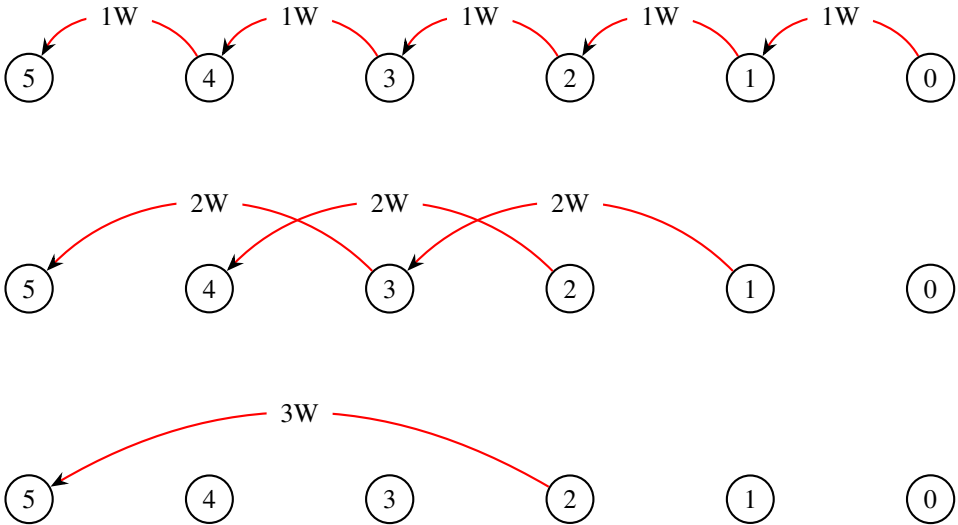
\begin{figure}
\vspace{1 cm}
\begin{center}
\begin{tikzpicture}[scale = 0.8]

\begin{scope}[every node/.style={circle,thick,draw}]  
    \node (0) at (6,0) {0};
    \node (1) at (3,0) {1};
    \node (2) at (0,0) {2};
    \node (3) at (-3,0) {3};
    \node (4) at (-6,0) {4};
    \node (5) at (-9,0) {5};
\end{scope}
\begin{scope}[>={Stealth[black]},
              every node/.style={fill=white,circle},
              every edge/.style={draw=red,thick}]
    \path [<-] (5) edge[bend left=60] node {1W} (4);
    \path [<-] (4) edge[bend left=60] node {1W} (3);
    \path [<-] (3) edge[bend left=60] node {1W} (2);
    \path [<-] (2) edge[bend left=60] node {1W} (1);
    \path [<-] (1) edge[bend left=60] node {1W} (0);
\end{scope}
\end{tikzpicture}

\begin{tikzpicture}[scale = 0.8]
\begin{scope}[every node/.style={circle,thick,draw}]  
    \node (0) at (6,0) {0};
    \node (1) at (3,0) {1};
    \node (2) at (0,0) {2};
    \node (3) at (-3,0) {3};
    \node (4) at (-6,0) {4};
    \node (5) at (-9,0) {5};
    \draw (0, 3) {};
\end{scope}
\begin{scope}[>={Stealth[black]},
              every node/.style={fill=white,circle},
              every edge/.style={draw=red,thick}]
    \path [<-] (5) edge[bend left=45] node {2W} (3);
    \path [<-] (4) edge[bend left=45] node {2W} (2);
    \path [<-] (3) edge[bend left=45] node {2W} (1);
\end{scope}
\end{tikzpicture}

\begin{tikzpicture}[scale = 0.8]
\begin{scope}[every node/.style={circle,thick,draw}]  
    \node (0) at (6,0) {0};
    \node (1) at (3,0) {1};
    \node (2) at (0,0) {2};
    \node (3) at (-3,0) {3};
    \node (4) at (-6,0) {4};
    \node (5) at (-9,0) {5};
    \draw (0, 3) {};
\end{scope}
\begin{scope}[>={Stealth[black]},
              every node/.style={fill=white,circle},
              every edge/.style={draw=red,thick}]
    \path [<-] (5) edge[bend left=30] node {3W} (2);
\end{scope}
\end{tikzpicture}


\caption{The remaining Right winning bids from Figure~\ref{fig:LWTB5} when dominated bids have been erased, for $\tb = 5$. These bids are relevant, even when property $\U$ holds.}\label{fig:RWDTB5}
\end{center}
\end{figure}

\clearpage
\bibliographystyle{ACM-Reference-Format}
\bibliography{EC-bibliography}

\end{document}